\documentclass[final,1p,times,authoryear]{elsarticle}
\usepackage{amsmath}
\usepackage{amsthm}
\usepackage{amssymb}
\usepackage{amsfonts}
\usepackage{dsfont}

\usepackage{float}
\usepackage[plainpages=false,pdfpagelabels=true,colorlinks=true,citecolor=blue,hypertexnames=false]{hyperref}
\usepackage[ruled,vlined]{algorithm2e}
\usepackage{mathtools}
\usepackage{extarrows}
\usepackage{url}
\usepackage{balance}
\usepackage{multirow}
\usepackage{color}
\usepackage{diagbox}
\usepackage{booktabs}
\usepackage{makecell}
\newcolumntype{V}{!{\vrule width 1pt}}

\newtheorem{theorem}{Theorem}
\newtheorem{lemma}[theorem]{Lemma}
\newtheorem{corollary}[theorem]{Corollary}

\newtheorem{proposition}[theorem]{Proposition}
\newtheorem{definition}[theorem]{Definition}

\newtheorem{example}[theorem]{Example}
\newtheorem{remark}[theorem]{Remark}
\newtheorem{problem}[theorem]{Problem}

\newcommand{\gr}{Gr\"{o}bner }
\def\z{{\bf z}}
\def\F{{\bf F}}
\def\G{{\bf G}}

\begin{document}

\begin{frontmatter}

\title{New Remarks on the Factorization and Equivalence Problems for a Class of Multivariate Polynomial Matrices}

\author[baic,smbh]{Dong Lu}
\ead{donglu@buaa.edu.cn}

\author[klmm,ucas]{Dingkang Wang}
\ead{dwang@mmrc.iss.ac.cn}

\author[zjnu]{Fanghui Xiao\corref{cor1}}
\ead{xiaofanghui@amss.ac.cn}

\cortext[cor1]{Corresponding author}

\address[baic]{Beijing Advanced Innovation Center for Big Data and Brain Computing, Beihang University, Beijing 100191, China}

\address[smbh]{School of Mathematical Sciences, Beihang University, Beijing 100191, China}

\address[klmm]{KLMM, Academy of Mathematics and Systems Science, Chinese Academy of Sciences, Beijing 100190, China}

\address[ucas]{School of Mathematical Sciences, University of Chinese Academy of Sciences, Beijing 100049, China}

\address[zjnu]{College of Mathematics and Computer Science, Zhejiang Normal University, Jinhua, 321004, China}

\begin{abstract}
 This paper is concerned with the factorization and equivalence problems of multivariate polynomial matrices. We present some new criteria for the existence of matrix factorizations for a class of multivariate polynomial matrices, and obtain a necessary and sufficient condition for the equivalence of a square polynomial matrix and a diagonal matrix. Based on the constructive proof of the new criteria, we give a factorization algorithm and prove the uniqueness of the factorization. We implement the algorithm on Maple, and two illustrative examples are given to show the effectiveness of the algorithm.
\end{abstract}

\end{frontmatter}

\section{Introduction}\label{sec_intro}

 Multidimensional systems have wide applications in image, signal processing, control of networked systems, and other areas (see, e.g., \cite{Bose1982,Bose2003}). A multidimensional system may be represented by a multivariate polynomial matrix, and we can obtain some important properties of the system by studying the corresponding matrix. Symbolic computation provides many effective theories and algorithms, such as module theory and \gr basis algorithm \citep{Cox2005Using,Lin2008ATutorial}, for the research of multidimensional systems. Therefore, the factorization and equivalence problems related to multivariate polynomial matrices have made great progress over the past decades.

 Up to now, the factorization problem for univariate and bivariate polynomial matrices has been completely solved by \cite{Morf1977New,Guiver1982Polynomial,Liu2013New}, but the case of more than two variables is still open. \cite{Youla1979Notes} first introduced three important concepts according to different properties of multivariate polynomial matrices, namely zero prime matrix factorization, minor prime matrix factorization and factor prime matrix factorization. When multivariate polynomial matrices satisfy several special properties, there are some results about the existence problem of zero prime matrix factorizations for the matrices (see, e.g., \cite{Charoenlarpnopparut1999Multidimensional,Lin1999Notes,Lin2001Further}).
 After that, \cite{Lin2001A} proposed the famous Lin-Bose conjecture: a multivariate polynomial matrix admits a zero prime matrix factorization if all its maximal reduced minors generate a unit ideal. This conjecture was proved by \cite{Pommaret2001Solving,Srinivas2004A,Wang2004On,Liu2014The}, respectively. \cite{Mingsheng2005On} gave a necessary and sufficient condition for a multivariate polynomial matrix with full rank to have a minor prime matrix factorization. They extracted an algorithm from Pommaret's proof of the Lin-Bose conjecture, and examples showed the effectiveness of the algorithm. \cite{Guan2019} generalized the main results in \cite{Mingsheng2005On} to the case of multivariate polynomial matrices without full rank. For the existence problem of factor prime matrix factorizations for multivariate polynomial matrices with full rank, \cite{Mingsheng2007On} and \cite{Liu2010Notes} introduced the concept of regularity and obtained a necessary and sufficient condition. \cite{Guan2018} gave an algorithm to judge whether a multivariate polynomial matrix with the greatest common divisor of all its maximal minors being square-free has a factor prime matrix factorization. However, the existence problem for factor prime matrix factorizations of multivariate polynomial matrices remains a challenging open problem so far.

 Comparing to the factorization problem of multivariate polynomial matrices which has been widely investigated during the past years, less attention has been paid to the equivalence problem of multivariate polynomial matrices. For any given multidimensional system, our goal is to simplify it into a simpler equivalent form.

 Since a univariate polynomial ring is a principal ideal domain, a univariate polynomial matrix is always equivalent to its Smith form. This implies that the equivalence problem of univariate polynomial matrices has been solved (see, e.g., \cite{Rosenbrock1970,Kailath1980}). For any given bivariate polynomial matrix, conditions under which it is equivalent to its Smith form have been investigated by \cite{Frost1978Equivalence,Lee1983Smith,Frost1986Some}. Note that the equivalence problem of two multivariate polynomial matrices is equivalent to the isomorphism problem of two finitely presented modules. \cite{Boudellioua2010Serre} and \cite{Cluzeau2008F,Cluzeau2013Is,Cluzeau2015A} obtained some important results by using module theory and homological algebra. According to the works of \cite{Boudellioua2010Serre}, \cite{Boudellioua2012Com,Boudellioua2014} designed some algorithms based on Maple to compute Smith forms for some classes of multivariate polynomial matrices. For the case of multivariate polynomial matrices with more than one variable, however, the equivalence problem is not yet fully solved due to the lack of a mature polynomial matrix theory (see, e.g., \cite{Kung1977New,Morf1977New,Pugh1998A}).

 From our personal viewpoint, new ideas need to be injected into these areas to obtain new theoretical results and effective algorithms. Therefore, it would be significant to provide some new criteria to study the factorization and equivalence problems for some classes of multivariate polynomial matrices.

 From the 1990s to the present, there is a class of multivariate polynomial matrices that has always attracted attention. That is,
 $$\mathcal{M} = \{\mathbf{F}\in k[\z]^{l\times m} : z_1 - f(\z_2) \text{ is a divisor of } d_l(\mathbf{F}) \text{ with } f(\z_2)\in k[\z_2]\},$$
 where $l \leq m$, $\z = \{z_1,\ldots,z_n\}$ with $n\geq 3$, $\z_2 = \{z_2,\ldots,z_n\}$ and $d_l(\mathbf{F})$ is the greatest common divisor of all the $l\times l$ minors of $\mathbf{F}$. People tried to solve the factorization and equivalence problems of multivariate polynomial matrices in $\mathcal{M}$. Let $\mathbf{F}\in \mathcal{M}$ and $h = z_1 - f(\z_2)$. Many factorization criteria on the existence of a matrix factorization for $\mathbf{F}$ with respect to $h$ have been proposed (see, e.g., \cite{Lin2001Factorizations,Liu2011On,Lu2019Factorizations}).  When $l=m$ and ${\rm det}(\mathbf{F})=h$, \cite{Lin2006On} proved that $\mathbf{F}$ is equivalent to its Smith form. After that, \cite{Li2017On} studied the equivalence problem of a square matrix $\F$ with ${\rm det}(\mathbf{F})=h^r$ and a diagonal matrix, where $r \geq 1$.

 Through research, there are still many multivariate polynomial matrices in $\mathcal{M}$ without satisfying previous factorization criteria or equivalence conditions, but they can be factorized with respect to $h$ or equivalent to simpler forms. As a consequence, we continue to study the factorization and equivalence problems of multivariate polynomial matrices in $\mathcal{M}$.

 This paper is an extension of \cite{Lu2020Further}, and the contributions listed following are new. 1) Under the assumption that $h$ is not a divisor of the greatest common divisor of all the $(l-1)\times (l-1)$ minors of $\F$, we give a necessary and sufficient condition for the existence of a matrix factorization of $\F$ with respect to $h$. 2) We summarize all factorization criteria for the existence of a matrix factorization of $\F$ with respect to $h$, and study the relationships among them. 3) For the case that $h$ is a divisor of the greatest common divisor of all the $(l-1)\times (l-1)$ minors of $\F$, we obtain a sufficient condition for the existence of a matrix factorization of $\F$ with respect to $h^r$, where $2 \leq r \leq l$. 4) Based on the new factorization criteria, we construct a new factorization algorithm and implement it on Maple; codes and examples are available on the website: \url{http://www.mmrc.iss.ac.cn/~dwang/software.html}.

 The rest of the paper is organized as follows. After a brief introduction to matrix factorization and matrix equivalence in section \ref{sec_notions}, we use two examples to propose two problems that we shall consider. We present in section \ref{sec_matrix factorization} two criteria for factorizing $\mathbf{F}$ with respect to $h$, and then study the relationships among all existed factorization criteria. A necessary and sufficient condition for the equivalence of a square polynomial matrix and a diagonal matrix is described in section \ref{sec_matrix equivalence}. We in section \ref{sec_generalization} generalize the main result in section \ref{sec_matrix factorization} to a more general case. We in section \ref{sec_algorithm} construct a factorization algorithm and study the uniqueness of matrix factorizations by the algorithm, and use two examples to illustrate the effectiveness of the algorithm in section \ref{sec_examples}. The paper contains a summary of contributions and some remarks in section \ref{sec_conclusions}.

\section{Preliminaries and problems} \label{sec_notions}

 In this section we first recall some basic notions which will be used in the following sections, and then we use two examples to put forward two problems that we are considering.

 \subsection{Basic notions}

 We denote by $k$ an algebraically closed field. Let $k[\z]$ and $k[\z_2]$ be the polynomial ring in variables $\z$ and $\z_2$ with coefficients in $k$, respectively. Let $k[\z]^{l\times m}$ be the set of $l\times m$ matrices with entries in $k[\z]$. Throughout the paper, we assume that $l\leq m$, and use uppercase bold letters to denote polynomial matrices. In addition, ``w.r.t." stands for ``with respect to".

 Let $\mathbf{F}\in k[\z]^{l\times m}$, we use $d_i(\mathbf{F})$ to denote the greatest common divisor of all the $i\times i$ minors of $\mathbf{F}$ with the convention that $d_0(\mathbf{F})=1$, where $i=1,\ldots,l$. Let $f\in k[\z_2]$, then $\mathbf{F}(f,\z_2)$ denotes a polynomial matrix in $k[\z_2]^{l\times m}$ which is formed by transforming $z_1$ in $\mathbf{F}$ into $f$.

 \begin{definition}[\cite{Lin1988On,Sule1994Feed}]
  Let $\mathbf{F}\in k[\z]^{l\times m}$ with rank $r$, where $1 \leq r \leq l$. For any given integer $i$ with $1\leq i \leq r$, let $a_1,\ldots,a_\beta$ denote all the $i\times i$ minors of $\mathbf{F}$, where $\beta = \binom l{i} \cdot \binom m{i}$. Extracting $d_i(\mathbf{F})$ from $a_1,\ldots,a_\beta$ yields
  $$a_j = d_i(\mathbf{F})\cdot b_j, ~~j=1,\ldots,\beta,$$
  then $b_1,\ldots,b_\beta$ are called the $i\times i$ reduced minors of $\mathbf{F}$.
 \end{definition}

 \cite{Lin1988On} showed that reduced minors are important invariants for polynomial matrices.

 \begin{lemma}\label{RM_relation}
  Let $\F_1\in k[\z]^{r\times t}$ be of full row rank, $b_1, \ldots, b_{\gamma}$ be all the $r\times r$ reduced minors of $\F_1$, and $\F_2\in k[\z]^{t\times (t-r)}$ be of full column rank, $\bar{b}_1, \ldots, \bar{b}_{\gamma}$ be all the $(t-r)\times (t-r)$ reduced minors of $\F_2$, where $r<t$ and $\gamma = \binom {t}{r}$. If $\F_1\F_2 = \mathbf{0}_{r\times(t-r)}$, then $\bar{b}_i=\pm b_i$ for $i=1,\ldots,\gamma$, and signs depend on indices.
 \end{lemma}

 Let $\mathbf{F}\in k[\z]^{l\times m}$ with rank $r$, where $1\leq r \leq l$. Let $\bar{\F}_1,\ldots,\bar{\F}_\eta \in k[\z]^{l\times r}$ be all the full column rank submatrices of $\F$, where $1\leq \eta \leq \binom{m}{r}$. According to Lemma \ref{RM_relation}, it is easy to prove that $\bar{\F}_1,\ldots,\bar{\F}_\eta$ have the same $r\times r$ reduced minors. Based on this phenomenon, we give the following concept.

 \begin{definition}
  Let $\mathbf{F}\in k[\z]^{l\times m}$ with rank $r$, and $\bar{\F}\in k[\z]^{l\times r}$ be an arbitrary full column rank submatrix of $\F$, where $1\leq r \leq l$. Let $c_1,\ldots, c_\xi$ be all the $r\times r$ reduced minors of $\bar{\F}$, where $\xi = \binom{l}{r}$. Then $c_1,\ldots, c_\xi$ are called the $r\times r$ {\bf column} reduced minors of $\F$.
 \end{definition}

  We can define the $r\times r$ {\bf row} reduced minors of $\F$ in the same way.

 In order to state conveniently problems and main conclusions of this paper, we introduce the following concepts and results.

 \begin{definition}
  Let $\mathbf{F}\in k[\z]^{l\times m}$ be of full row rank.
  \begin{enumerate}
    \item If all the $l\times l$ minors of $\mathbf{F}$ generate $k[\z]$, then $\mathbf{F}$ is said to be a zero left prime (ZLP) matrix.

    \item If all the $l\times l$ minors of $\mathbf{F}$ are relatively prime, i.e., $d_l(\mathbf{F})$ is a nonzero constant in $k$, then $\mathbf{F}$ is said to be a minor left prime (MLP) matrix.

    \item If for any polynomial matrix factorization $\mathbf{F} = \mathbf{F}_1\mathbf{F}_2$ with $\mathbf{F}_1\in k[\z]^{l\times l}$, $\mathbf{F}_1$ is necessarily a unimodular matrix, i.e., ${\rm det}(\mathbf{F}_1)$ is a nonzero constant in $k$, then $\mathbf{F}$ is said to be a factor left prime (FLP) matrix.
  \end{enumerate}
 \end{definition}

 Zero right prime (ZRP) matrices, minor right prime (MRP) matrices and factor right prime (FRP) matrices can be similarly defined for matrices $\mathbf{F}\in k[\z]^{m\times l}$ with $m\geq l$. We refer to \cite{Youla1979Notes} for more details about the relationships among ZLP matrices, MLP matrices and FLP matrices.

 For any given ZLP matrix $\mathbf{F}\in k[\z]^{l\times m}$, \cite{Quillen1976Projective} and \cite{Suslin1976Projective} proved that an $m\times m$ unimodular matrix can be constructed such that $\mathbf{F}$ is its first $l$ rows, respectively. This result is called Quillen-Suslin theorem, and it solved the problem raised by \cite{serre1955faisceaux}.

\begin{lemma}\label{QS-2}
 If $\mathbf{F}\in k[\z]^{l \times m}$ is a ZLP matrix, then a unimodular matrix $\mathbf{U}\in k[\z]^{m\times m}$ can be constructed such that $\mathbf{F}$ is its first $l$ rows.
\end{lemma}

 There are many algorithms for the Quillen-Suslin theorem, we refer to \cite{Youla1984The,Logar1992Algorithms,park1995} for more details. \cite{Fabianska2007Applications} first designed a Maple package, which is called QUILLENSUSLIN, to implement the Quillen-Suslin theorem.

 Let $W$ be a $k[\z]$-module generated by $\vec{u}_1,\ldots,\vec{u}_l\in k[\z]^{1\times m}$. The set of all $(b_1,\ldots,b_l)\in k[\z]^{1\times l}$ such that $b_1\vec{u}_1+ \cdots + b_l\vec{u}_l = \vec{0}$ is a $k[\z]$-module of $k[\z]^{1\times l}$, is called the (first) syzygy module of $W$, and denoted by ${\rm Syz}(W)$. \cite{Lin1999On} proposed several interesting structural properties of syzygy modules. Let $\mathbf{F} = \begin{bmatrix} \vec{u}_1^{\rm T},\ldots,\vec{u}_l^{\rm T} \end{bmatrix}^{\rm T}$. The rank of $W$ is defined as the rank of $\mathbf{F}$ that is denoted by ${\rm rank}(\mathbf{F})$.  \cite{Guan2018} proved that the rank of $W$ does not depend on the choice of generators of $W$.

 \begin{lemma}\label{rank-syz}
  With above notations. If ${\rm rank}(W) = r$ with $1\leq r \leq l$, then the rank of ${\rm Syz}(W)$ is $l-r$.
 \end{lemma}
 \begin{proof}
   Let $k(\z)$ be the fraction field of $k[\z]$, and ${\rm Syz}^\ast(W) = \{ \vec{v}\in k(\z)^{1\times l} : \vec{v}\cdot\mathbf{F} = \vec{0}\}$. Then, ${\rm Syz}^\ast(W)$ is a $k(\z)$-vector space of dimension $l-r$. For any given $l-r+1$ different vectors $\vec{v}_1,\ldots,$ $\vec{v}_{l-r+1} \in k[\z]^{1\times l}$ in ${\rm Syz}(W)$, it is obvious that $\vec{v}_i \in {\rm Syz}^\ast(W)$ for each $i$, and they are $k(\z)$-linearly dependent. This implies that $\vec{v}_1,\ldots,\vec{v}_{l-r+1}$ are $k[\z]$-linearly dependent. Thus ${\rm rank}({\rm Syz}(W)) \leq l-r$.

   Assume that $\vec{p}_1,\ldots,\vec{p}_{l-r} \in k(\z)^{1\times l}$ are $l-r$ vectors in ${\rm Syz}^\ast(W)$, and they are $k(\z)$-linearly independent. For each $j$, we have $p_{j1}\vec{u}_1+ \cdots + p_{jl}\vec{u}_l = \vec{0}$, where $\vec{p}_j = (p_{j1},\ldots,p_{jl})$. Multiplying both sides of the equation by the least common multiple of the denominators of $p_{j1},\ldots,p_{jl}$, we obtain $\bar{p}_j = (\bar{p}_{j1},\ldots,\bar{p}_{jl})\in k[\z]$ such that $\bar{p}_{j1}\vec{u}_1+ \cdots + \bar{p}_{jl}\vec{u}_l = \vec{0}$. Then, $\bar{p}_j\in {\rm Syz}(W)$, where $j=1,\ldots, l-r$. Moreover, $\bar{p}_1,\ldots,\bar{p}_{l-r}$ are $k[\z]$-linearly independent. Thus, ${\rm rank}({\rm Syz}(W)) \geq l-r$.

   As a consequence, the rank of ${\rm Syz}(W)$ is $l-r$ and the proof is completed.
 \end{proof}

 \begin{remark}
  Assume that ${\rm Syz}(W)$ is generated by $\vec{v}_1,\ldots,\vec{v}_t\in k[\z]^{1\times l}$, and $\mathbf{H} = \begin{bmatrix} \vec{v}_1^{\rm T}, \ldots,\vec{v}_t^{\rm T} \end{bmatrix}^{\rm T}$. It follows from ${\rm rank}(\mathbf{H}) = l-r$ that $t \geq l-r$. That is, the number of vectors in any given generators of ${\rm Syz}(W)$ is greater than or equal to $l-r$.
 \end{remark}

 Let $\mathbf{F}\in k[\z]^{l\times m}$ with rank $r$, where $1\leq r \leq l$. For each $1\leq i \leq r$, we use $I_i(\mathbf{F})$ to denote the ideal generated by all the $i\times i$ minors of $\mathbf{F}$. For convenience, let $I_0(\mathbf{F}) = k[\z]$. Moreover, we denote the submodule of $k[\z]^{1\times m}$ generated by all the row vectors of $\mathbf{F}$ by ${\rm Im}(\mathbf{F})$.

\begin{definition}
 Let $W$ be a finitely generated $k[\z]$-module, and $k[\z]^{1\times l} \xlongrightarrow{\phi} k[\z]^{1\times m} \rightarrow W \rightarrow 0$ be a presentation of $W$, where $\phi$ acts on the right on row vectors, i.e., $\phi(\vec{u}) = \vec{u}\cdot\mathbf{F}$ for $\vec{u}\in k[\z]^{1\times l}$ with $\mathbf{F}$ being a presentation matrix corresponding to the linear mapping $\phi$. Then the ideal $Fitt_j(W) = I_{m-j}(\mathbf{F})$ is called the $j$-th Fitting ideal of $W$. Here, we make the convention that $Fitt_j(W) = k[\z]$ for $j \geq m$, and that $Fitt_j(W) = 0$ for $j < {\rm max}\{m-l,0\}$.
\end{definition}

 We remark that $Fitt_j(W)$ only depend on $W$ (see, e.g., \cite{Greuel2002A,Eisenbud2013}). In addition, the chain $0 = Fitt_{-1}(W) \subseteq Fitt_0(W) \subseteq \ldots \subseteq Fitt_m(W) = k[\z]$ of Fitting ideals is increasing. \cite{Cox2005Using} showed that one obtains the presentation matrix $\mathbf{F}$ for $W$ by arranging the generators of ${\rm Syz}(W)$ as rows.

\subsection{Matrix factorization problem}

 A matrix factorization of a multivariate polynomial matrix is formulated as follows.

 \begin{definition}\label{matrix_factorization}
  Let $\mathbf{F}\in k[\z]^{l\times m}$ and $h_0\mid d_l(\mathbf{F})$. $\mathbf{F}$ is said to admit a matrix factorization w.r.t. $h_0$ if $\mathbf{F}$ can be factorized as
  \begin{equation}\label{gerneral-matirx-factorization}
   \mathbf{F} = \mathbf{G}_1\mathbf{F}_1
  \end{equation}
  such that $\mathbf{G}_1\in k[\z]^{l\times l}$ with ${\rm det}(\mathbf{G}_1) = h_0$ and $\mathbf{F}_1\in k[\z]^{l\times m}$. In particular, Equation (\ref{gerneral-matirx-factorization}) is said to be a ZLP (MLP, FLP) matrix factorization if $\mathbf{F}_1$ is a ZLP (MLP, FLP) matrix.
 \end{definition}

 Throughout the paper, let $h= z_1 - f(\z_2)$ with $f(\z_2)\in k[\z_2]$. This paper will address the following specific matrix factorization problem.

 \begin{problem}\label{problem_1}
  Let $\mathbf{F}\in \mathcal{M}$. Under what conditions do $\mathbf{F}$ have a matrix factorization w.r.t. $h$.
 \end{problem}

 So far, several results have been made on Problem \ref{problem_1}, and the latest progress on this problem was obtained by \cite{Lu2019Factorizations}.

 \begin{lemma}\label{theorem_Lu_2019}
 Let $\mathbf{F}\in \mathcal{M}$. If $h \nmid d_{l-1}(\mathbf{F})$ and the ideal generated by $h$ and all the $(l-1)\times (l-1)$ reduced minors of $\mathbf{F}$ is $k[\z]$, then $\mathbf{F}$ admits a matrix factorization w.r.t. $h$.
\end{lemma}

 Although Lemma \ref{theorem_Lu_2019} gives a criterion to determine whether $\mathbf{F}$ has a matrix factorization w.r.t. $h$, we found that there exist some polynomial matrices in $\mathcal{M}$ which do not satisfy the conditions of Lemma \ref{theorem_Lu_2019}, but still admit matrix factorizations w.r.t. $h$.

 \begin{example}\label{counter-example-1}
  Let
  \[\mathbf{F} =
  \begin{bmatrix}
    -2z_1z_2^2+z_1^2z_3+z_2^2z_3-z_1z_3^2+z_2z_3^2 & z_1^3-z_2^3-z_1^2z_3+z_2z_3^2 & z_1z_2-z_2z_3 &  z_2^2   \\
    -z_1z_2+z_3^2  &  -z_2^2+z_1z_3  &  0  & z_2
  \end{bmatrix}\]
  be a polynomial matrix in $\mathbb{C}[z_1,z_2,z_3]^{2\times 4}$, where $\mathbb{C}$ is the complex field.

  It is easy to compute that $d_2(\mathbf{F})=z_2(z_1-z_3)$ and $d_1(\mathbf{F})=1$. Let $h=z_1-z_3$, then $h\mid d_2(\mathbf{F})$ implies that $\mathbf{F}\in \mathcal{M}$. Obviously, $h\nmid d_1(\mathbf{F})$. Since $d_1(\mathbf{F})=1$, the entries in $\mathbf{F}$ are all the $1\times 1$ reduced minors of $\mathbf{F}$. Let $\prec_{\z}$ be the degree reverse lexicographic order, then the reduced \gr basis $G$ of the ideal generated by $h$ and all the $1\times 1$ reduced minors of $\mathbf{F}$ w.r.t. $\prec_{\z}$ is $\{z_1-z_3,z_2,z_3^2\}$. It follows from $G \neq \{1\}$ that Lemma \ref{theorem_Lu_2019} cannot be applied.

  However, $\mathbf{F}$ admits a matrix factorization w.r.t. $h$, i.e., there exist $\mathbf{G}_1\in \mathbb{C}[z_1,z_2,z_3]^{2\times 2}$ and $\mathbf{F}_1\in \mathbb{C}[z_1,z_2,z_3]^{2\times 4}$ such that
  \[\mathbf{F} = \mathbf{G}_1 \mathbf{F}_1 =
    \begin{bmatrix}
    h &  z_2 \\
       0      & 1
   \end{bmatrix}
   \begin{bmatrix}
    z_1z_3-z_2^2   &   z_1^2-z_2z_3  & z_2 &  0   \\
    -z_1z_2+z_3^2  &  -z_2^2+z_1z_3  &  0  & z_2
   \end{bmatrix},\]
  where ${\rm det}(\mathbf{G}_1)= h$.
 \end{example}

 From the above example we see that Problem \ref{problem_1} is far from being resolved. So, in the next section we make a detailed analysis on this problem.

\subsection{Matrix equivalence problem}

 Now we introduce the concept of the equivalence of two multivariate polynomial matrices.

 \begin{definition}
  Two polynomial matrices $\mathbf{F}_1\in k[\z]^{l\times m}$ and $\mathbf{F}_2\in k[\z]^{l\times m}$ are said to be equivalent if there exist two unimodular matrices $\mathbf{U}\in k[\z]^{l\times l}$ and $\mathbf{V}\in k[\z]^{m\times m}$ such that
  \begin{equation}\label{smith_equ_3}
   \mathbf{F}_1 = \mathbf{U} \mathbf{F}_2 \mathbf{V}.
  \end{equation}
 \end{definition}

 In fact, a univariate polynomial matrix is always equivalent to its Smith form. However, this result is not valid for the case of more than one variable, and there are many counter-examples (see, e.g., \cite{Lee1983Smith,Boudellioua2013Further}). Hence, people began to consider under what conditions multivariate polynomial matrices in $k[\z]$ are equivalent to simpler forms. \cite{Li2017On} investigated the equivalence problem for a class of multivariate polynomial matrices and obtained the following result.

 \begin{lemma}\label{Li_equivalent}
  Let $\mathbf{F}\in k[\z]^{l\times l}$ with ${\rm det}(\mathbf{F}) = h^r$, where $h=z_1-f(\z_2)$ and $r$ is a positive integer. Then $\mathbf{F}$ is equivalent to ${\rm diag}(h^r,1,\ldots,1)$ if and only if $h^r$ and all the $(l-1)\times (l-1)$ minors of $\mathbf{F}$ generate $k[\z]$.
 \end{lemma}

 For a given square matrix that does not satisfy the condition of Lemma \ref{Li_equivalent}, we use the following example to illustrate that it can be equivalent to another diagonal matrix.

 \begin{example}\label{counter-example-2}
  Let
  \[\mathbf{F} =
  \begin{bmatrix}
    z_1z_2-z_2^2+z_2z_3+z_2-z_3-1 & z_1z_2z_3-z_2^2z_3+z_1z_2-z_2^2+z_2z_3-z_3 & z_1z_2z_3-z_2^2z_3   \\
    z_1z_2-z_2^2+z_1-z_2+z_3+1  &  (z_1-z_2)(z_2z_3+2z_2+z_3+1)+z_3  & \mathbf{F}[2,3] \\
    z_1-z_2 & z_1z_3-z_2z_3+2z_1-2z_2 & z_1z_3-z_2z_3+z_1-z_2
  \end{bmatrix}\]
  be a polynomial matrix in $\mathbb{C}[z_1,z_2,z_3]^{3\times 3}$, where $\mathbf{F}[2,3] = z_1z_2z_3-z_2^2z_3+z_1z_2-z_2^2+z_1z_3-z_2z_3$, and $\mathbb{C}$ is the complex field.

  It is easy to compute that ${\rm det}(\mathbf{F})=(z_1-z_2)^2$. Let $h = z_1-z_2$ and $\prec_{\z}$ be the degree reverse lexicographic order, then the reduced \gr basis $G$ of the ideal generated by $h^2$ and all the $2\times 2$ minors of $\mathbf{F}$ w.r.t. $\prec_{\z}$ is $\{z_1-z_2\}$. It follows from $G \neq \{1\}$ that Lemma \ref{Li_equivalent} cannot be applied.

  However, $\mathbf{F}$ is equivalent to ${\rm diag}(h,h,1)$, i.e., there exist two unimodular polynomial matrices $\mathbf{U}\in \mathbb{C}[z_1,z_2,z_3]^{3\times 3}$ and $\mathbf{V}\in \mathbb{C}[z_1,z_2,z_3]^{3\times 3}$ such that
  \[\mathbf{F} = \mathbf{U} \cdot {\rm diag}(h,h,1) \cdot \mathbf{V} =\begin{bmatrix}
      0   &  z_2   &  z_2 - 1  \\
     z_2  &  z_2+1 &      1    \\
      1   &    1   &      0
   \end{bmatrix}
   \begin{bmatrix}
    h &    0      &    0    \\
    0 &    h      &    0    \\
    0 &    0      &    1
   \end{bmatrix}
   \begin{bmatrix}
    0     &   1     &  1   \\
    1     & z_3+1   & z_3 \\
    z_3+1 &   z_3   & 0
   \end{bmatrix}.\]
 \end{example}

 Based on the phenomenon of Example \ref{counter-example-2}, we consider the following matrix equivalence problem in this paper.

 \begin{problem}\label{problem_2}
  Let $\mathbf{F}\in k[\z]^{l\times l}$ with ${\rm det}(\mathbf{F}) = h^r$, where $h = z_1-f(\z_2)$ and $1 \leq r \leq l$. What is the sufficient and necessary condition for the equivalence of $\mathbf{F}$ and ${\rm diag}(\underbrace{h,\ldots,h}_{r}, \underbrace{1,\ldots,1}_{l-r})$?
 \end{problem}

\section{Factorization for polynomial matrices}\label{sec_matrix factorization}

 In this section, we first propose two criteria to judge whether $\mathbf{F}\in \mathcal{M}$ has a matrix factorization w.r.t. $h$, and then study the relationships among all existed factorization criteria.

\subsection{A sufficient condition}

 We first introduce two lemmas.

\begin{lemma}[\cite{Wang2004On}]\label{Lin-Bose Con}
 Let $\mathbf{F}\in k[\z]^{l\times m}$ with rank $r$, and all the $r\times r$ reduced minors of $\mathbf{F}$ generate $k[\z]$. Then there exist $\mathbf{G}_1\in k[\z]^{l\times r}$ and $\mathbf{F}_1\in k[\z]^{r\times m}$ such that $\mathbf{F} = \mathbf{G}_1\mathbf{F}_1$ with $\mathbf{F}_1$ being a ZLP matrix.
\end{lemma}

\begin{lemma}[\cite{Lin2001Factorizations}] \label{lemma_zero}
 Let $p\in k[\z]$ and $f(\z_2)\in k[\z_2]$. Then $z_1 - f(\z_2)$ is a divisor of $p$ if and only if $p(f,\z_2)$ is a zero polynomial in $k[\z_2]$.
\end{lemma}

 Now, we propose a sufficient condition to factorize $\F$ w.r.t. $h$.

 \begin{theorem}\label{main-theorem-problem-1}
  Let $\mathbf{F}\in \mathcal{M}$ and $W = {\rm Im}(\mathbf{F}(f,\z_2))$. If $Fitt_{l-2}(W)=0$ and $Fitt_{l-1}(W)=\langle d \rangle$ with $d\in k[\z_2]\setminus \{0\}$, then $\mathbf{F}$ admits a matrix factorization w.r.t. $h$.
 \end{theorem}
 \begin{proof}
  Let $k[\z_2]^{1\times s} \xlongrightarrow{\phi} k[\z_2]^{1\times l} \rightarrow W \rightarrow 0$ be a presentation of $W$, and $\mathbf{H}\in k[\z_2]^{s\times l}$ be a matrix corresponding to the linear mapping $\phi$. Then ${\rm Syz}(W) = {\rm Im}(\mathbf{H})$.

  It follows from $Fitt_{l-2}(W)=0$ that all the $2\times 2$ minors of $\mathbf{H}$ are zero polynomials. Then, ${\rm rank}(\mathbf{H}) \leq 1$. Moreover, $Fitt_{l-1}(W)=\langle d \rangle$ with $d\in k[\z_2]\setminus \{0\}$ implies that ${\rm rank}(\mathbf{H}) \geq 1$. As a consequence, we have ${\rm rank}(\mathbf{H}) = 1$.

  Let $a_1,\ldots,a_\beta \in k[\z_2]$ and $b_1,\ldots,b_\beta \in k[\z_2]$ be all the $1\times 1$ minors and $1\times 1$ reduced minors of $\mathbf{H}$, respectively. Then, $a_i = d_1(\mathbf{H})\cdot b_i$ for $i=1,\ldots, \beta$. Since $\langle a_1,\ldots,a_\beta \rangle = \langle d \rangle$, it is obvious that $d \mid d_1(\mathbf{H})$. Moreover, we have $d = \sum_{i=1}^{\beta}c_i a_i$ for some $c_i\in k[\z_2]$. Thus $d = d_1(\mathbf{H})\cdot ( \sum_{i=1}^{\beta}c_i b_i)$. This implies that $d_1(\mathbf{H}) \mid d$. Hence $d = \delta\cdot d_1(\mathbf{H})$, where $\delta$ is a nonzero constant. Therefore, $\langle b_1,\ldots,b_\beta \rangle = k[\z_2]$.

  According to Lemma \ref{Lin-Bose Con}, there exist $\vec{u}\in k[\z_2]^{s\times 1}$ and $\vec{w}\in k[\z_2]^{1\times l}$ such that $\mathbf{H} = \vec{u}\vec{w}$ with $\vec{w}$ being a ZLP vector. It follows from ${\rm Syz}(W) = {\rm Im}(\mathbf{H})$ that $\vec{u}\vec{w}\mathbf{F}(f,\z_2) = \mathbf{0}_{s\times m}$. Since $\vec{u}$ is a column vector, we have $\vec{w}\mathbf{F}(f,\z_2) = \mathbf{0}_{1\times m}$.

 Using the Quillen-Suslin theorem, we can construct a unimodular matrix $\mathbf{U}\in k[\z_2]^{l\times l}$ such that $\vec{w}$ is its first row. Let $\mathbf{F}_0 = \mathbf{U}\mathbf{F}$, then the first row of $\mathbf{F}_0(f,\z_2)= \mathbf{U}\mathbf{F}(f,\z_2)$ is zero vector. By Lemma \ref{lemma_zero}, $h$ is a common divisor of the polynomials in the first row of $\mathbf{F}_0$, thus
 \[\mathbf{F}_0=\mathbf{U}\mathbf{F} = \mathbf{D} \mathbf{F}_1={\rm diag}(h, \underbrace{1,\ldots,1}_{l-1})\cdot
   \begin{bmatrix}
     \bar{f}_{11}  & \bar{f}_{12}& \cdots   & \bar{f}_{1m}  \\
         \vdots    &     \vdots  & \vdots   &  \vdots    \\
     \bar{f}_{l1}  & \bar{f}_{l2}& \cdots   & \bar{f}_{lm}
   \end{bmatrix}.
 \]
 Consequently, we can now derive the matrix factorization of $\mathbf{F}$ w.r.t. $h$, i.e., $\mathbf{F}= \mathbf{G}_1\mathbf{F}_1$, where $\mathbf{G}_1= \mathbf{U}^{-1}\mathbf{D} \in k[\z]^{l\times l}$, $\mathbf{F}_1\in k[\z]^{l\times m}$ and ${\rm det}(\mathbf{G}_1)=h$.
 \end{proof}

 \subsection{A necessary and sufficient condition for a special case}

 In Theorem \ref{main-theorem-problem-1}, the conditions $Fitt_{l-2}(W)$ $=0$ and $Fitt_{l-1}(W)=\langle d \rangle$ imply that the rank of $\F(f,\z_2)$ is $l-1$. In the following, we first give a lemma about  the necessary and sufficient condition for ${\rm rank}(\F(f,\z_2)) = l-1$.

 \begin{lemma}\label{rank_matrix}
  Let $\mathbf{F}\in \mathcal{M}$. Then ${\rm rank}(\F(f,\z_2)) = l-1$ if and only if $h\nmid d_{l-1}(\F)$.
 \end{lemma}

 \begin{proof}
  Since $h\mid d_l(\F)$, we have ${\rm rank}(\F(f,\z_2)) \leq l-1$. Let $a_1,\ldots,a_\gamma \in k[\z]$ be all the $(l-1)\times (l-1)$ minors of $\F$, then $a_1(f,\z_2),\ldots,a_\gamma(f,\z_2)$ are all the $(l-1)\times (l-1)$ minors of $\F(f,\z_2)$.

  Assume that ${\rm rank}(\F(f,\z_2)) = l-1$, then there is at least one integer $i$ with $1\leq i \leq \gamma$ such that $a_i(f,\z_2)$ is a nonzero polynomial. According to Lemma \ref{lemma_zero}, $h$ is not a divisor of $a_i$. Obviously, $h\nmid d_{l-1}(\F)$.

  Suppose $h\nmid d_{l-1}(\F)$. If ${\rm rank}(\F(f,\z_2)) < l-1$, then $a_j(f,\z_2) = 0$, $j=1,\ldots,\gamma$. This implies that $h$ is a common divisor of $a_1,\ldots,a_\gamma$, which leads to a contradiction. Therefore, ${\rm rank}(\F(f,\z_2)) = l-1$.
 \end{proof}

 \begin{lemma}[\cite{Lin2005On}]\label{lin-2005-lemma}
  Let $\mathbf{G}\in k[\z]^{l\times l}$ with ${\rm det}(\mathbf{G})=h$, then there is a ZLP vector $\vec{w}\in k[\z_2]^{1\times l}$ such that $\vec{w}\mathbf{G}(f,\z_2) = \mathbf{0}_{1\times l}$.
 \end{lemma}

 Now, we give a partial solution to Problem \ref{problem_1}.

 \begin{theorem}\label{main-theorem-problem-2}
  Let $\mathbf{F}\in \mathcal{M}$ with $h\nmid d_{l-1}(\F)$. Then the following are equivalent:
  \begin{enumerate}
    \item $\mathbf{F}$ admits a matrix factorization w.r.t. $h$;
    \item all the $(l-1)\times (l-1)$ column reduced minors of $\F(f,\z_2)$ generate $k[\z_2]$.
  \end{enumerate}
 \end{theorem}

 \begin{proof}
  $1 \rightarrow 2$. If $\mathbf{F}$ admits a matrix factorization w.r.t. $h$, then there are $\mathbf{G}_1\in k[\z]^{l\times l}$ and $\mathbf{F}_1\in k[\z]^{l\times m}$ such that $\mathbf{F}=\mathbf{G}_1\mathbf{F}_1$ with ${\rm det}(\mathbf{G}_1)=h$. Obviously, $\mathbf{F}(f,\z_2)=\mathbf{G}_1(f,\z_2)\mathbf{F}_1(f,\z_2)$. Since ${\rm det}(\mathbf{G}_1)=h$, by Lemma \ref{lin-2005-lemma} there is a ZLP vector $\vec{w}\in k[\z_2]^{1\times l}$ such that $\vec{w}\mathbf{G}_1(f,\z_2) = \mathbf{0}_{1\times l}$. This implies that $\vec{w}\mathbf{F}(f,\z_2)=\mathbf{0}_{1\times m}$. According to Lemma \ref{rank_matrix}, we have ${\rm rank}(\mathbf{F}(f,\z_2)) = l-1$. Using Lemma \ref{RM_relation}, all the $(l-1)\times (l-1)$ column reduced minors of $\F(f,\z_2)$ are equivalent to all the $1\times 1$ reduced minors of $\vec{w}$. It follows that all the $(l-1)\times (l-1)$ column reduced minors of $\F(f,\z_2)$ generate $k[\z_2]$.

  $2 \rightarrow 1$. Sine ${\rm rank}(\mathbf{F}(f,\z_2)) = l-1$, there is a nonzero vector $\vec{w}=[w_1,\ldots,w_l] \in k[\z_2]^{1\times l}$ such that $\vec{w}\mathbf{F}(f,\z_2) = \mathbf{0}_{1\times m}$. As all the $(l-1)\times (l-1)$ column reduced minors of $\F(f,\z_2)$ generate $k[\z_2]$, all the $1\times 1$ reduced minors of $\vec{w}$ generate $k[\z_2]$ by Lemma \ref{RM_relation}. Assume that $w_0 \in k[\z_2]$ is the greatest common divisor of $w_1,\ldots,w_l$, then $\vec{w}/w_0$ is a ZLP vector. Using Quillen-Suslin theorem, we can construct a unimodular matrix $\mathbf{U}\in k[\z_2]^{l\times l}$ such that $\vec{w}/w_0$ is its first row. This implies that there are $\mathbf{D}\in k[\z]^{l\times l}$ and $\mathbf{F}_1\in k[\z]^{l\times m}$ such that $\mathbf{U}\mathbf{F} = \mathbf{D}\mathbf{F}_1$, where $\mathbf{D} = {\rm diag}(h,1,\ldots,1)$. Therefore, we obtain a matrix factorization of $\mathbf{F}$ w.r.t. $h$, i.e., $\mathbf{F}=\mathbf{G}_1 \mathbf{F}_1$, where $\mathbf{G}_1 = \mathbf{U}^{-1}\mathbf{D}$ and ${\rm det}(\mathbf{G}_1)=h$.
 \end{proof}

 \subsection{Comparison among all existed factorization criteria}

 Let $\F \in \mathcal{M}$, and $a_1,\ldots,a_\beta\in k[\z]$ be all the $l\times l$ minors of $\F$. Since $h\mid d_l(\F)$, there are $e_1,\ldots,e_\beta\in k[\z]$ such that $a_i = h e_i$, $i=1,\ldots,\beta$. \cite{Lin2001Factorizations} proved that $\F$ has a matrix factorization w.r.t. $h$ if $\langle h,e_1,\ldots,e_\beta \rangle = k[\z]$. The main idea is as follows. $\langle h,e_1,\ldots,e_\beta \rangle = k[\z]$ implies that ${\rm rank}(\F(f,\z_2)) = l-1$ for every $\z_2 \in k^{n-1}$, then we can construct a ZLP vector $\vec{w}\in k[\z_2]^{1\times l}$ such that $\vec{w} \F(f,\z_2) = \mathbf{0}_{1\times m}$. Obviously, in this situation, we have $h\nmid d_{l-1}(\F)$. So, the condition $\langle h,e_1,\ldots,e_\beta \rangle = k[\z]$ is a special case of Theorem \ref{main-theorem-problem-2}.

 When $d_l(\F) = h$, \cite{Lin2005On} proved that $\F$ has an MLP matrix factorization w.r.t. $h$ if and only if all the $(l-1)\times (l-1)$ column reduced minors of $\F(f,\z_2)$ generate $k[\z_2]$. In fact, $d_l(\F) = h$ implies that $h \nmid d_{l-1}(\F)$. Hence, the main result of \cite{Lin2005On} is also a special case of Theorem \ref{main-theorem-problem-2}.

 Let $c_1,\ldots,c_\eta\in k[\z]$ be all the $(l-1)\times (l-1)$ minors of $\F$. \cite{Liu2011On} proved that ${\rm rank}(\F(f,\z_2)) = l-1$ for every $\z_2 \in k^{n-1}$ if and only if $\langle h,c_1,\ldots,c_\eta \rangle = k[\z]$. Then, $\F$ has a matrix factorization w.r.t. $h$ if $\langle h,c_1,\ldots,c_\eta \rangle = k[\z]$. Although \cite{Liu2011On} generalized the main result of \cite{Lin2001Factorizations}, $\langle h,c_1,\ldots,c_\eta \rangle = k[\z]$ is still a special case of Theorem \ref{main-theorem-problem-2}.

 Let $b_1,\ldots,b_\eta\in k[\z]$ be all the $(l-1)\times (l-1)$ reduced minors of $\F$. \cite{Lu2019Factorizations} proved that $\F$ has a matrix factorization w.r.t. $h$ if $h\nmid d_{l-1}(\F)$ and $\langle h,b_1,\ldots,b_\eta \rangle = k[\z]$. We explain the difference between $\langle h,c_1,\ldots,c_\eta \rangle = k[\z]$ and $\langle h,b_1,\ldots,b_\eta \rangle = k[\z]$. $\langle h,c_1,\ldots,c_\eta \rangle = k[\z]$ implies that all the $(l-1)\times (l-1)$ minors of $\F(f,\z_2)$ generate $k[\z_2]$, and $\langle h,b_1,\ldots,b_\eta \rangle = k[\z]$ implies that all the $(l-1)\times (l-1)$ reduced minors of $\F(f,\z_2)$ generate $k[\z_2]$. Therefore, the main result of \cite{Lu2019Factorizations} is a generalization of that of \cite{Liu2011On}. Under the assumption that $h\nmid d_{l-1}(\F)$, there is no doubt that $\langle h,b_1,\ldots,b_\eta \rangle = k[\z]$ is a special case of Theorem \ref{main-theorem-problem-2}.

 Assume that $h\nmid d_{l-1}(\F)$ and $\langle h,b_1,\ldots,b_\eta \rangle = k[\z]$, then there exists a ZLP vector $\vec{w}\in k[\z_2]^{1\times l}$ such that $\vec{w} \F(f,\z_2) = \mathbf{0}_{1\times m}$. In Theorem \ref{main-theorem-problem-1}, $Fitt_{l-2}(W)=0$ and $Fitt_{l-1}(W)=\langle d \rangle$ implies that all the $1\times 1$ reduced minors of $\mathbf{H}$ generate $k[\z_2]$. Then we can obtain a ZLP vector $\vec{w}\in k[\z_2]^{1\times l}$ by factorizing $\mathbf{H}$. Although the conditions in \cite{Lu2019Factorizations} and Theorem \ref{main-theorem-problem-1} all imply that we can construct a ZLP vector $\vec{w}\in k[\z_2]^{1\times l}$, $\langle h,b_1,\ldots,b_\eta \rangle = k[\z]$ cannot deduce $Fitt_{l-1}(W)=\langle d \rangle$. It follows that Theorem \ref{main-theorem-problem-1} is not a generalization of the main result in \cite{Lu2019Factorizations}. However, Example \ref{counter-example-1} shows that Theorem \ref{main-theorem-problem-1} can solve some problems that the main result in \cite{Lu2019Factorizations} cannot solve.

 Assume that $\mathbf{H}\in k[\z_2]^{s\times l}$ is composed of a system of generators of the syzygy module of $\F(f,\z_2)$. Then, ${\rm Syz}(\F(f,\z_2)) = {\rm Im}(\mathbf{H})$. $Fitt_{l-1}(W)=\langle d \rangle$ in Theorem \ref{main-theorem-problem-1} implies that all the $1\times 1$ reduced minors of $\mathbf{H}$ generate $k[\z_2]$. According to Lemma \ref{RM_relation}, all the $(l-1)\times (l-1)$ column reduced minors of $\F(f,\z_2)$ generate $k[\z_2]$. Thus, Theorem \ref{main-theorem-problem-1} can deduce Theorem \ref{main-theorem-problem-2}. However, Theorem \ref{main-theorem-problem-2} only imply that all the $1\times 1$ row reduced minors of $\mathbf{H}$ generate $k[\z_2]$. It follows that Theorem \ref{main-theorem-problem-1} is not equivalent to Theorem \ref{main-theorem-problem-2}. Therefore, Theorem \ref{main-theorem-problem-1} is a special case of Theorem \ref{main-theorem-problem-2}.

 Based on Lemma \ref{Lin-Bose Con}, \cite{Liu2013New} proposed a criterion for the existence of a matrix factorization of $\F$ w.r.t. $h_0$.

 \begin{lemma}\label{Liu-Wang-2013}
  Let $\F \in k[\z]^{l\times m}$ be a full row rank matrix, and $h_0\in k[\z]$ be a divisor of $d_l(\F)$. $a_1,\ldots,a_\beta\in k[\z]$ and $c_1,\ldots,c_\eta\in k[\z]$ be all the $l\times l$ minors and  $(l-1)\times (l-1)$ minors of $\F$, respectively. There are $e_1,\ldots,e_\beta\in k[\z]$ such that $a_i = h_0 e_i$, $i=1,\ldots,\beta$. If $h_0,e_1,\ldots,e_\beta,c_1,\ldots,c_\eta$ generate $k[\z]$, then $\F$ has a matrix factorization w.r.t. $h_0$.
 \end{lemma}

 In Lemma \ref{Liu-Wang-2013}, $\F$ does not have to belong to $\mathcal{M}$ and $h_0$ does not have to be of the form $z_1-f(\z_2)$. Obviously, the main results of \cite{Lin2001Factorizations} and \cite{Liu2011On} are special cases of Lemma \ref{Liu-Wang-2013}. When $h_0 = z_1 - f(\z_2)$, however, we find that $\langle h_0,e_1,\ldots,e_\beta,c_1,\ldots,c_\eta \rangle = k[\z]$ is equivalent to $\langle h_0,c_1,\ldots,c_\eta \rangle = k[\z]$. That is, Lemma \ref{Liu-Wang-2013} is the same as the main result of \cite{Liu2011On} for the case of $h_0 = z_1 - f(\z_2)$. Before proving this conclusion, we first introduce a lemma which proposed by \cite{Lin2001Factorizations}.

 \begin{lemma}\label{Lin-2001-rank}
  Let $\F \in k[z_1]^{l\times m}$ be a univariate polynomial matrix with full row rank, and $d\in k[z_1]$ be the greatest common divisor of all the $l\times l$ minors of $\F$. If $z_{11}\in k$ is a simple zero of $d$, i.e., $z_1 - z_{11}$ is a divisor of $d$, but $(z_1 - z_{11})^2$ is not a divisor of $d$, then ${\rm rank}(\F(z_{11})) = l-1$.
 \end{lemma}

 Now, we can assert that the following conclusion is correct.

 \begin{proposition}\label{Liu-Wang-equal-Liu}
  Let $\F \in \mathcal{M}$, $a_1,\ldots,a_\beta\in k[\z]$ and $c_1,\ldots,c_\eta\in k[\z]$ be all the $l\times l$ minors and  $(l-1)\times (l-1)$ minors of $\F$, respectively. There are $e_1,\ldots,e_\beta\in k[\z]$ such that $a_i = h e_i$, $i=1,\ldots,\beta$. Then, $\langle h,e_1,\ldots,e_\beta,c_1,\ldots,c_\eta \rangle = k[\z]$ if and only if $\langle h,c_1,\ldots,c_\eta \rangle = k[\z]$.
 \end{proposition}

 \begin{proof}
  Sufficiency is obvious, we next prove the necessity.

  Assume that $\langle h,e_1,\ldots,e_\beta,c_1,\ldots,c_\eta \rangle = k[\z]$. If $\langle h,c_1,\ldots,c_\eta \rangle \neq k[\z]$, then there exists a point $\vec{\varepsilon} = (\varepsilon_1, \ldots, \varepsilon_n)\in k^n$ such that
  $$ \varepsilon_1 = f(\varepsilon_2, \ldots, \varepsilon_n) ~ \text{ and } ~ c_i(\vec{\varepsilon}) = 0, ~ i=1,\ldots,\eta.$$
  Then, ${\rm rank}(F(\vec{\varepsilon})) < l-1$. Let $\tilde{\F} = \F(z_1, \varepsilon_2, \ldots, \varepsilon_n)$ be a univariate polynomial matrix with entries in $k[z_1]$, and $\tilde{a}_1,\ldots,\tilde{a}_\beta\in k[z_1]$ be all the $l\times l$ minors of $\tilde{\F}$. Obviously, we have
  $$ \tilde{a}_j = a_j(z_1, \varepsilon_2, \ldots, \varepsilon_n) = (z_1 - \varepsilon_1)\cdot e_j(z_1, \varepsilon_2, \ldots, \varepsilon_n), ~ j = 1,\ldots,\beta. $$
  Assume that $q\in k[z_1]$ is the greatest common divisor of $e_1(z_1, \varepsilon_2, \ldots, \varepsilon_n),\ldots,e_\beta(z_1, \varepsilon_2, \ldots, \varepsilon_n)$, then $d_l(\tilde{\F}) = (z_1 - \varepsilon_1)\cdot q$. It follows from $\langle h,e_1,\ldots,e_\beta,c_1,\ldots,c_\eta \rangle = k[\z]$ that $\vec{\varepsilon}$ is not a common zero of the system $\{ e_1 =0,\ldots,e_\beta =0\}$. Thus, $\varepsilon_1$ is not a zero of $p$. This implies that $\varepsilon_1$ is a simple zero of $d_l(\tilde{\F})$. According to Lemma \ref{Liu-Wang-equal-Liu}, we have ${\rm rank}(\tilde{\F}(\varepsilon_1)) = l-1$, which leads to a contradiction. Therefore, $\langle h,c_1,\ldots,c_\eta \rangle = k[\z]$.
 \end{proof}

 \section{Equivalence for polynomial matrices}\label{sec_matrix equivalence}

 In this section, we first put forward a necessary and sufficient condition to solve Problem \ref{problem_2}, and then use an example to illustrate the effectiveness of the matrix equivalence theorem.

 We introduce a lemma, which is called the Binet-Cauchy formula \citep{Strang2010Linear}.

 \begin{lemma}\label{binet-cauchy}
 Let $\mathbf{F}=\mathbf{G}_1\mathbf{F}_1$, where $\mathbf{G}_1\in k[\z]^{l\times l}$ and $\mathbf{F}_1 \in k[\z]^{l\times m}$. Then an $i\times i$ minor of $\mathbf{F}$ is
  \begin{equation*}
  {\rm det}\Bigl(\mathbf{F}\begin{pmatrix}\begin{smallmatrix}
  r_1\cdots r_i \\ j_1\cdots j_i \end{smallmatrix}\end{pmatrix}\Bigr)\\
      =  \sum_{1\leq s_1<\cdots<s_i\leq l}{\rm det}\Bigl(\mathbf{G}_1\begin{pmatrix}\begin{smallmatrix}
  r_1\cdots r_i \\ s_1\cdots s_i
  \end{smallmatrix}\end{pmatrix}\Bigr)\cdot {\rm det}\Bigl(\mathbf{F}_1\begin{pmatrix}\begin{smallmatrix}
  s_1\cdots s_i \\ j_1\cdots j_i\end{smallmatrix}\end{pmatrix}\Bigr).
  \end{equation*}
\end{lemma}

 In Lemma \ref{binet-cauchy}, $\mathbf{F}\begin{pmatrix}\begin{smallmatrix} r_1\cdots r_i \\ j_1\cdots j_i \end{smallmatrix}\end{pmatrix}$ denotes an $i\times i$ submatrix consisting of the $r_1,\ldots,r_i$ rows and $j_1,\ldots,j_i$ columns of $\mathbf{F}$. Based on this lemma, we can obtain the following two results.

 \begin{lemma}\label{gcd-equal}
  Let $\mathbf{F}\in k[\z]^{l\times m}$ be of full row rank with $\mathbf{F}=\mathbf{G}_1\mathbf{F}_1$, where $\mathbf{G}_1\in k[\z]^{l\times l}$ and $\mathbf{F}_1 \in k[\z]^{l\times m}$. Then $d_i(\mathbf{F}_1)\mid d_i(\mathbf{F})$ and $d_i(\mathbf{G}_1)\mid d_i(\mathbf{F})$ for each $i\in\{1,\ldots,l\}$.
 \end{lemma}

 \begin{proof}
  We only prove $d_i(\mathbf{F}_1)\mid d_i(\mathbf{F})$, since the proof of $d_i(\mathbf{G}_1)\mid d_i(\mathbf{F})$ follows in a similar manner. For any given $i\in\{1,\ldots,l\}$, let $a_{i,1},\ldots,a_{i,t_i}$ and $\bar{a}_{i,1},\ldots,\bar{a}_{i,t_i}$ be all the $i\times i$ minors of $\mathbf{F}$ and $\mathbf{F}_1$ respectively, where $t_i = \binom l{i}\binom m{i}$. For each $a_{i,j}$, it is a $k[\z]$-linear combination of $\bar{a}_{i,1},\ldots,\bar{a}_{i,t_i}$ by using Lemma \ref{binet-cauchy}, where $j=1,\ldots,t_i$. Since $d_i(\mathbf{F}_1)$ is the greatest common divisor of $\bar{a}_{i,1},\ldots, \bar{a}_{i,t_i}$, for each $j$ we have $d_i(\mathbf{F}_1) \mid a_{i,j}$. Then, $d_i(\mathbf{F}_1)\mid d_i(\mathbf{F})$.
 \end{proof}

 \begin{lemma}\label{gcd-equal-cgd}
  Let $\mathbf{F}_1,\mathbf{F}_2\in k[\z]^{l\times m}$ be of full row rank. If $\mathbf{F}_1$ and $\mathbf{F}_2$ are equivalent, then $d_i(\mathbf{F}_1) = d_i(\mathbf{F}_2)$ for each $i\in\{1,\ldots,l\}$.
 \end{lemma}

 \begin{proof}
  Since $\mathbf{F}_1$ and $\mathbf{F}_2$ are equivalent, then there exist two unimodular matrices $\mathbf{U}\in k[\z]^{l\times l}$ and $\mathbf{V}\in k[\z]^{m\times m}$ such that $\mathbf{F}_1 = \mathbf{U} \mathbf{F}_2 \mathbf{V}$. For each $i\in\{1,\ldots,l\}$, it follows from Lemma \ref{gcd-equal} that $d_i(\mathbf{F}_2)\mid d_i(\mathbf{U} \mathbf{F}_2) \mid  d_i(\mathbf{F}_1)$. Furthermore, we have $\mathbf{F}_2 = \mathbf{U}^{-1} \mathbf{F}_1 \mathbf{V}^{-1}$ since $\mathbf{U}$ and $\mathbf{V}$ are two unimodular matrices. Similarly, we obtain $d_i(\mathbf{F}_1)\mid d_i(\mathbf{U}^{-1} \mathbf{F}_1) \mid  d_i(\mathbf{F}_2)$. Therefore, $d_i(\mathbf{F}_1) = d_i(\mathbf{F}_2)$ up to multiplication by a nonzero constant.
 \end{proof}

\begin{lemma}[\cite{Lu2017}]\label{theorem_Lu_2017}
 Let $\mathbf{F}\in k[\z]^{l\times m}$ with rank $l-r$. If all the $(l-r)\times (l-r)$ minors of $\mathbf{F}$ generate $k[\z]$, then there exists a ZLP matrix $\mathbf{H}\in k[\z]^{r\times l}$ such that $\mathbf{H}\mathbf{F}=\mathbf{0}_{r\times m}$.
\end{lemma}

 Combining Lemma \ref{theorem_Lu_2017} and the Quillen-Suslin theorem, we can now solve Problem \ref{problem_2}.

\begin{theorem}\label{matrix-equivalent}
 Let $\mathbf{F}\in k[\z]^{l\times l}$ with ${\rm det}(\mathbf{F})=h^r$, where $h = z_1-f(\z_2)$ and $1 \leq r \leq l$. Then $\mathbf{F}$ and ${\rm diag}(\underbrace{h,\ldots,h}_{r}, \underbrace{1,\ldots,1}_{l-r})$ are equivalent if and only if $h\mid d_{l-r+1}(\mathbf{F})$ and the ideal generated by $h$ and all the $(l-r)\times (l-r)$ minors of $\mathbf{F}$ is $k[\z]$.
\end{theorem}

\begin{proof}
 For convenience, let $\mathbf{D} = {\rm diag}(h,\ldots,h,1,\ldots,1)$ and $\bar{\mathbf{F}} = \mathbf{F}(f,\z_2)$. Let $a_1,\ldots, a_\beta$ be all the $(l-r)\times (l-r)$ minors of $\mathbf{F}$. It is obvious that $a_1(f,\z_2),\ldots, a_\beta(f,\z_2)$ are all the $(l-r)\times (l-r)$ minors of $\bar{\mathbf{F}}$.

 Sufficiency. It follows from $h\mid d_{l-r+1}(\mathbf{F})$ that ${\rm rank}(\bar{\mathbf{F}}) \leq l-r$. Assume that there exists a point $(\varepsilon_2,\ldots,\varepsilon_n) \in k^{1\times (n-1)}$ such that
 \begin{equation}\label{equ-main-theorem-1}
  a_i(f(\varepsilon_2,\ldots,\varepsilon_n),\varepsilon_2,\ldots,
 \varepsilon_n)=0,~~i =1,\ldots,\beta.
 \end{equation}
 Let $\varepsilon_1 = f(\varepsilon_2,\ldots,\varepsilon_n)$, then Equation (\ref{equ-main-theorem-1}) implies that $(\varepsilon_1,\varepsilon_2,\ldots,$ $\varepsilon_n) \in k^{1\times n}$ is a common zero of the polynomial system $\{ h=0, a_1 =0,\ldots, a_\beta=0\}$. This contradicts the fact that $h$ and all the $(l-r)\times (l-r)$ minors of $\mathbf{F}$ generate $k[\z]$. Then, all the $(l-r)\times (l-r)$ minors of $\bar{\mathbf{F}}$ generate $k[\z_2]$. According to Lemma \ref{theorem_Lu_2017}, there exists a ZLP matrix $\mathbf{H}\in k[\z_2]^{r\times l}$ such that $\mathbf{H}\bar{\mathbf{F}} = \mathbf{0}_{r\times l}$. Based on the Quillen-Suslin theorem, we can construct a unimodular matrix $\mathbf{U}\in k[\z_2]^{l\times l}$ such that $\mathbf{H}$ is its first $r$ rows. Then, there is a polynomial matrix $\mathbf{V}\in k[\z]^{l\times l}$ such that $\mathbf{U} \mathbf{F} = \mathbf{D} \mathbf{V}$. Since ${\rm det}(\mathbf{F})=h^r$ and $\mathbf{U}$ is a unimodular matrix, we have $\mathbf{F} = \mathbf{U}^{-1} \mathbf{D} \mathbf{V}$ and $\mathbf{V}$ is a unimodular matrix. Therefore, $\mathbf{F}$ and $\mathbf{D}$ are equivalent.

 Necessity. If $\mathbf{F}$ and $\mathbf{D}$ are equivalent, then there exist two unimodular matrices $\mathbf{U}\in k[\z]^{l\times l}$ and $\mathbf{V}\in k[\z]^{l\times l}$ such that $\mathbf{F}= \mathbf{U}  \mathbf{D} \mathbf{V}$. It follows from Lemma \ref{gcd-equal-cgd} that $d_{l-r+1}(\mathbf{F}) = d_{l-r+1}(\mathbf{D}) = h$. If $\langle h, a_1,\ldots,a_\beta \rangle \neq k[\z]$, then there exists a point $\vec{\varepsilon}\in k^{1\times n}$ such that $h(\vec{\varepsilon})=0$ and ${\rm rank}(\mathbf{F}(\vec{\varepsilon}))< l-r$. Obviously, ${\rm rank}(\mathbf{D}(\vec{\varepsilon})) = l-r$ and ${\rm rank}(\mathbf{U}^{-1}(\vec{\varepsilon})) = {\rm rank}(\mathbf{V}^{-1}(\vec{\varepsilon})) = l$. Since $\mathbf{D}= \mathbf{U}^{-1} \mathbf{F}\mathbf{V}^{-1}$, we have
 $${\rm rank}(\mathbf{D}(\vec{\varepsilon})) \leq {\rm min}\{{\rm rank}(\mathbf{U}^{-1}(\vec{\varepsilon})), {\rm rank}(\mathbf{F}(\vec{\varepsilon})), {\rm rank}(\mathbf{V}^{-1}(\vec{\varepsilon}))\},$$
 which leads to a contradiction. Therefore, $\langle h, a_1,\ldots,a_\beta \rangle = k[\z]$ and the proof is completed.
\end{proof}

\begin{remark}
  When $r=l$ in Theorem \ref{matrix-equivalent}, we just need to check whether $h$ is a divisor of $d_1(\mathbf{F})$.
\end{remark}

 Now, we use Example \ref{counter-example-2} to illustrate a constructive method which follows the proof process of the sufficiency of Theorem \ref{matrix-equivalent} and explain how to obtain the two unimodular matrices associated with equivalent matrices.

 \begin{example}
  Let $\mathbf{F}$ be the same polynomial matrix as in Example \ref{counter-example-2}. It is easy to compute that ${\rm det}(\mathbf{F})=(z_1-z_2)^2$ and $d_2(\mathbf{F}) = z_1-z_2$. Let $h = z_1-z_2$, it is obvious that $h\mid  d_2(\mathbf{F})$. The reduced \gr basis of the ideal generated by $h$ and all the $1\times 1$ minors of $\mathbf{F}$ w.r.t. $\prec_{\z}$ is $\{1\}$. Then, $\mathbf{F}$ is equivalent to ${\rm diag}(h,h,1)$.

  Note that
  \[\mathbf{F}(z_2,z_2,z_3) =
  \begin{bmatrix}
   (z_3+1)(z_2-1)  & z_3(z_2-1)  & 0    \\
       z_3+1    &     z_3     &    0    \\
       0    &  0  &  0
   \end{bmatrix},\]
  ${\rm rank}(\mathbf{F}(z_2,z_2,z_3)) = 1$. Let $W = {\rm Im}(\mathbf{F}(z_2,z_2,z_3))$. We compute a system of generators of the syzygy module of $W$, and obtain
  \[\mathbf{H} =
  \begin{bmatrix}
    1   & -z_2+1  &  z_2^2-z_2    \\
   -1   & z_2-1   &  -z_2^2+z_2+1
  \end{bmatrix}\]
  such that $\mathbf{H}\cdot \mathbf{F}(z_2,z_2,z_3) = \mathbf{0}_{2\times 3}$. It is easy to check that $\mathbf{H}$ is a ZLP matrix. Then, a unimodular matrix $\mathbf{U}\in k[\z_2]^{3\times 3}$ can be constructed such that $\mathbf{H}$ is its the first $2$ rows by using the Maple package QUILLENSUSLIN, where
  \[\mathbf{U} =
  \begin{bmatrix}
        1   & -z_2+1  &  z_2^2-z_2    \\
       -1   & z_2-1   &  -z_2^2+z_2+1 \\
       -1   & z_2     &  -z_2^2
  \end{bmatrix}.\]
  Now we can extract $h$ from the first $2$ rows of $\mathbf{U} \mathbf{F}$, and get
  \[\mathbf{F} = \mathbf{U}^{-1}\cdot {\rm diag}(h,h,1)\cdot \mathbf{V} =\begin{bmatrix}
      0   &  z_2   &  z_2 - 1  \\
     z_2  &  z_2+1 &      1    \\
      1   &    1   &      0
   \end{bmatrix}
   \begin{bmatrix}
    h &    0      &    0    \\
    0 &    h      &    0    \\
    0 &    0      &    1
   \end{bmatrix}
   \begin{bmatrix}
    0     &   1     &  1   \\
    1     & z_3+1   & z_3 \\
    z_3+1 &   z_3   & 0
   \end{bmatrix}.\]
 \end{example}

 \section{Generalizations}\label{sec_generalization}

 We construct the following two sets of polynomial matrices:
 $$ \mathcal{M}_1 = \{ \F \in \mathcal{M} : h \nmid d_{l-1}(\F)\} ~ \text{ and } ~ \mathcal{M}_2 = \{ \F \in \mathcal{M} : h \mid d_{l-1}(\F)\}.$$
 Let $\F\in \mathcal{M}$. Assume that $h = z_1 - f(\z_2)$ is given, then $\F \in \mathcal{M}_1$ or $\F \in \mathcal{M}_2$. If $\F \in \mathcal{M}_1$, we can use Theorem \ref{main-theorem-problem-2} to judge whether $\F$ has a matrix factorization w.r.t. $h$. If $\F \in \mathcal{M}_2$, we need to propose some criteria to factorize $\F$.

 Since $d_0(\mathbf{F}) \mid d_1(\mathbf{F}) \mid \cdots \mid d_{l-1}(\mathbf{F}) \mid d_l(\mathbf{F})$, there exists a unique integer $r$ with $1\leq r \leq l$ such that $h \mid d_{l-r+1}(\mathbf{F})$ but $h \nmid d_{l-r}(\mathbf{F})$. Based on this fact, we subdivide $\mathcal{M}_2$ into the following sets:
 $$ \mathcal{M}_{2,r} = \{ \F \in \mathcal{M}_2 : h \mid d_{l-r+1}(\mathbf{F}) ~ \text{ but } ~ h \nmid d_{l-r}(\mathbf{F})\}, ~ r=2,\ldots,l.$$

 \begin{lemma}\label{rank_matrix-2}
  Let $\mathbf{F}\in \mathcal{M}_2$. Then ${\rm rank}(\F(f,\z_2)) = l-r$ with $2 \leq r \leq l$ if and only if $\mathbf{F}\in \mathcal{M}_{2,r}$.
 \end{lemma}

 The proof of Lemma \ref{rank_matrix-2} is basically the same as that of Lemma \ref{rank_matrix}, so it is omitted here. Inspired by Theorem \ref{main-theorem-problem-2} and Theorem \ref{matrix-equivalent}, we propose the following result for the existence of a matrix factorization of $\mathbf{F} \in \mathcal{M}_{2,r}$ w.r.t. $h^r$, where $2 \leq r < l$.

 \begin{theorem}\label{main-theorem-1-2-2}
  Let $\mathbf{F}\in \mathcal{M}_{2,r}$ with $2 \leq r < l$, then the following are equivalent:
  \begin{enumerate}
    \item there are $\mathbf{G}_1\in k[\z]^{l\times l}$ and $\mathbf{F}_1\in k[\z]^{l\times m}$ such that $\mathbf{F}=\mathbf{G}_1\mathbf{F}_1$, and $\mathbf{G}_1$ is equivalent to ${\rm diag}(\underbrace{h,\ldots,h}_{r}, \underbrace{1,\ldots,1}_{l-r})$;

    \item all the $(l-r)\times (l-r)$ column reduced minors of $\F(f,\z_2)$ generate $k[\z_2]$.
  \end{enumerate}
 \end{theorem}

 \begin{proof}
  $1 \rightarrow 2$. Since $\mathbf{G}_1$ and ${\rm diag}(h,\ldots,h,1,\ldots,1)$ are equivalent, we have $h\mid d_{l-r+1}(\mathbf{G}_1)$ and $\langle h, g_1,\ldots,g_\eta \rangle$ $= k[\z]$ by Theorem \ref{matrix-equivalent}, where $g_1,\ldots,g_\eta$ are all the $(l-r)\times (l-r)$ minors of $\mathbf{G}_1$. This implies that all the $(l-r)\times (l-r)$ minors of $\mathbf{G}_1(f,\z_2)$ generate $k[\z_2]$. According to Lemma \ref{theorem_Lu_2017}, we can construct a ZLP matrix $\mathbf{W}\in k[\z_2]^{r\times l}$ such that $\mathbf{W}\mathbf{G}_1(f,\z_2) = \mathbf{0}_{r\times l}$. It follows from $\mathbf{F}=\mathbf{G}_1\mathbf{F}_1$ that $\mathbf{W}\mathbf{F}(f,\z_2)=\mathbf{0}_{r\times m}$. Since $\mathbf{W}$ is a ZLP matrix, all the $(l-r)\times (l-r)$ column reduced minors of $\F(f,\z_2)$ generate $k[\z_2]$.

  $2 \rightarrow 1$. From Lemma \ref{rank_matrix-2}, there exists a full row rank matrix $\mathbf{H}\in k[\z_2]^{r\times l}$ such that $\mathbf{H}\F(f,\z_2) = \mathbf{0}_{r\times m}$. Since all the $(l-r)\times (l-r)$ column reduced minors of $\F(f,\z_2)$ generate $k[\z_2]$, all the $r\times r$ reduced minors of $\mathbf{H}$ generate $k[\z_2]$ by Lemma \ref{RM_relation}. Using Lemma \ref{Lin-Bose Con}, $\mathbf{H}$ has a ZLP matrix factorization: $\mathbf{H} = \mathbf{H}_1\mathbf{H}_2$, where $\mathbf{H}_1\in k[\z_2]^{r\times r}$, and $\mathbf{H}_2\in k[\z_2]^{r\times l}$ is a ZLP matrix. As $\mathbf{H}_1$ is a full column rank matrix, it follows from $\mathbf{H}\F(f,\z_2) = \mathbf{0}_{r\times m}$ that $\mathbf{H}_2\F(f,\z_2) = \mathbf{0}_{r\times m}$. Using Quillen-Suslin theorem, we can construct a unimodular matrix $\mathbf{U}\in k[\z_2]^{l\times l}$ such that $\mathbf{H}_2$ is its first $r$ rows. This implies that there is $\mathbf{F}_1\in k[\z]^{l\times m}$ such that $\mathbf{U}\mathbf{F} = \mathbf{D}\mathbf{F}_1$, where $\mathbf{D} = {\rm diag}(h,\ldots,h,1,\ldots,1)$ with ${\rm det}(\mathbf{D})=h^r$. Therefore, we obtain a matrix factorization of $\mathbf{F}$ w.r.t. $h^r$, i.e., $\mathbf{F}=\mathbf{G}_1 \mathbf{F}_1$ with $\mathbf{G}_1 = \mathbf{U}^{-1}\mathbf{D}$. Obviously, $\mathbf{G}_1$ is equivalent to $\mathbf{D}$.
 \end{proof}

 \begin{remark}\label{main-theorem-2-remark}
  In Theorem \ref{main-theorem-1-2-2}, the matrix factorization $\F = \G_1\F_1$ must satisfies that $\mathbf{G}_1$ is equivalent to ${\rm diag}(h,\ldots,h, 1,\ldots,1)$. Since there exist many polynomial matrices such that their matrix factorizations do not satisfy this requirement, the condition ``all the $(l-r)\times (l-r)$ column reduced minors of $\F(f,\z_2)$ generate $k[\z_2]$" is only a sufficient condition for the existence of a matrix factorization of $\F \in \mathcal{M}_{2,r}$ w.r.t. $h^r$, where $2 \leq r < l$.
 \end{remark}

 \begin{theorem}\label{main-theorem-remark}
  Let $\mathbf{F}\in \mathcal{M}_{2,l}$, then $h$ is a common divisor of all entries in $\F$. We can extract $h$ from each row of $\mathbf{F}$ and obtain a matrix factorization of $\mathbf{F}$ w.r.t. $h^l$.
 \end{theorem}

 Let $k[\bar{\z}_i] = k[z_1,\ldots,z_{i-1},z_{i+1},\ldots,z_n]$ and $h_i = z_i - f(\bar{\z}_i)$, where $f(\bar{\z}_i)\in k[\bar{\z}_i]$ and $1 \leq i \leq n$. We construct the following sets of polynomial matrices:
 $$\mathcal{M}^{(i,r)} = \{\mathbf{F}\in k[\z]^{l\times m} : h_i \mid d_{l-r+1}(\mathbf{F}) ~ \text{ but } ~ h_i \nmid d_{l-r}(\mathbf{F})\}, ~ r=1,\ldots,l.$$
 Then, we can get the following corollary.

 \begin{corollary}\label{corollary-problem-1}
  Let $\mathbf{F}\in \mathcal{M}^{(i,r)}$, where $1\leq i \leq n$ and $1 \leq r \leq l$. If all the $(l-r)\times (l-r)$ column reduced minors of $\F(z_1,\ldots,z_{i-1},f,z_{i+1},\ldots,z_n)$ generate $k[\bar{\z}_i]$, then $\F$ admits a matrix factorization w.r.t. $h_i^r$.
 \end{corollary}

\section{Factorization algorithm and its uniqueness}\label{sec_algorithm}

 In this section, we first propose an algorithm to factorize $\F\in \mathcal{M}$ w.r.t. $h^r$, where $1 \leq r \leq l$. And then, we study the uniqueness of matrix factorizations by the algorithm.

\subsection{Factorization algorithm}

 According to Theorem \ref{main-theorem-problem-2}, Theorem \ref{main-theorem-1-2-2} and Theorem \ref{main-theorem-remark}, we construct an algorithm to factorize polynomial matrices in $\mathcal{M}$.

 \begin{algorithm}[!htb]
 \DontPrintSemicolon
 \SetAlgoSkip{}
 \LinesNumbered
 \SetKwInOut{Input}{Input}
 \SetKwInOut{Output}{Output}

 \Input{$\mathbf{F}\in \mathcal{M}$, $h=z_1-f(\z_2)$ and a monomial order $\prec_{\z_2}$ in $k[\z_2]$.}

 \Output{a matrix factorization of $\mathbf{F}$ w.r.t. $h^r$, where $1 \leq r \leq l$.}

 \Begin{

  compute the rank $l-r$ of $\mathbf{F}(f,\z_2)$;

  \If{$r = l$}
  {
    extract $h$ from each row of $\mathbf{F}$ and obtain $\mathbf{F}_1$, i.e., $\mathbf{F} = {\rm diag}(h,\ldots,h) \cdot \mathbf{F}_1$;

    {\bf return} ${\rm diag}(h,\ldots,h)$ and $\mathbf{F}_1$.
  }

  compute a reduced \gr basis $\mathcal{G}$ of all the $(l-r)\times (l-r)$ column reduced minors of $\mathbf{F}(f,\z_2)$ w.r.t. $\prec_{\z_2}$;

  \If{$\mathcal{G} \neq \{1\}$}
  {
    \If{$r =1$}
    {
      {\bf return} $\F$ has no matrix factorizations w.r.t. $h$.
    }
    \Else
    {
      {\bf return} unable to judge.
    }
  }

  compute a ZLP matrix $\mathbf{H} \in k[\z_2]^{r\times l}$ such that $\mathbf{H}\mathbf{F}(f,\z_2) = \mathbf{0}_{r\times m}$;

  construct a unimodular matrix $\mathbf{U}\in k[\z_2]^{l\times l}$ such that $\mathbf{H}$ is its first $r$ rows;

  compute $\F_1 \in k[\z]^{l\times m}$ such that $\mathbf{U}\mathbf{F} = {\rm diag}(h,\ldots,h,1,\ldots,1) \cdot \mathbf{F}_1$;

  {\bf return} $\mathbf{U}^{-1}\cdot {\rm diag}(h,\ldots,h,1,\ldots,1)$ and $\mathbf{F}_1$.
 }
 \caption{factorization algorithm}
 \label{MF_Algorithm}
 \end{algorithm}

 \begin{theorem}
  Algorithm \ref{MF_Algorithm} works correctly.
 \end{theorem}

 \begin{proof}
  The proof follows directly from Theorem \ref{main-theorem-problem-2}, Theorem \ref{main-theorem-1-2-2} and Theorem \ref{main-theorem-remark}.
 \end{proof}

 \noindent Before proceeding further, let us remark on Algorithm \ref{MF_Algorithm}.
 \begin{itemize}
  \item It follows from $\mathcal{G} \neq \{1\}$ in Step 7 that all the $(l-r)\times (l-r)$ column reduced minors of $\mathbf{F}(f,\z_2)$ do not generate $k[\z_2]$.

  \item Under the assumption that $\mathcal{G} \neq \{1\}$ and $r = 1$, the algorithm in Step 9 returns that ``$\F$ has no matrix factorizations w.r.t. $h$" by Theorem \ref{main-theorem-problem-2}. When $\mathcal{G} \neq \{1\}$ and $1 < r < l$, the algorithm in Step 11 returns that ``unable to judge" by Remark \ref{main-theorem-2-remark}.

  \item We explain how to calculate a ZLP matrix $\mathbf{H}$ in Step 12. We first compute a \gr basis $\mathcal{G}^\ast$ of the syzygy module of $\mathbf{F}(f,\z_2)$. As ${\rm rank}(\mathbf{F}(f,\z_2)) = l-r$, we can select $r$ $k[\z_2]$-linearly independent vectors from $\mathcal{G}^\ast$ and form $\mathbf{H}_0 \in k[\z_2]^{r\times l}$ with full row rank. According to Lemma \ref{RM_relation}, all the $r\times r$ reduced minors of $\mathbf{H}_0$ generate $k[\z_2]$. Then, $\mathbf{H}_0$ has a ZLP matrix factorization by Lemma \ref{Lin-Bose Con}. Hence, we second use the Maple package QUILLENSUSLIN to compute a ZLP matrix factorization of $\mathbf{H}_0$ and obtain a ZLP matrix $\mathbf{H}$.

  \item In Step 13 we use QUILLENSUSLIN again to construct a unimodular matrix. Since QUILLENSUSLIN is a Maple package, we implement the factorization algorithm on Maple. Codes and examples are available on the website: \url{http://www.mmrc.iss.ac.cn/~dwang/software.html}.
 \end{itemize}

\subsection{Uniqueness of matrix factorizations}

 \cite{Liu2015Further} studied the uniqueness problem of polynomial matrix factorizations. They pointed out that for a non-regular factor $h_0$ of $\mathbf{F}\in k[\z]^{l\times m}$, under the condition that there exists a matrix factorization $\mathbf{F}=\mathbf{G}_1\mathbf{F}_1$ with ${\rm det}(\mathbf{G}_1) = h_0$, ${\rm Im}(\mathbf{F}_1)$ is not uniquely determined. In other words, when $\mathbf{F}=\mathbf{G}_1\mathbf{F}_1 = \mathbf{G}_2\mathbf{F}_2$ with ${\rm det}(\mathbf{G}_1) = {\rm det}(\mathbf{G}_2) = h_0$, ${\rm Im}(\mathbf{F}_1)$ and ${\rm Im}(\mathbf{F}_2)$ might not be the same.

 Let $\mathbf{F}\in \mathcal{M}$, $h = z_1 - f(\z_2)$ and $\prec_{\z_2}$ are given. We use Algorithm \ref{MF_Algorithm} to factorize $\F$ w.r.t. $h^r$, where $1\leq r \leq l$. Assume that all the $(l-r)\times (l-r)$ column reduced minors of $\F(f,\z_2)$ generate $k[\z_2]$, then we need to compute a ZLP matrix and construct a unimodular matrix. Due to the different choices of a ZLP matrix and a unimodular matrix, we will get different matrix factorizations of $\mathbf{F}$ w.r.t. $h^r$. Hence, in the following we study the uniqueness of matrix factorizations by Algorithm \ref{MF_Algorithm}.

 \begin{theorem}\label{main-theorem-uniqueness}
  Let $\mathbf{F}\in \mathcal{M}$ satisfy $\mathbf{F}= \mathbf{U}_1^{-1}\mathbf{D}\mathbf{F}_1 = \mathbf{U}_2^{-1}\mathbf{D}\mathbf{F}_2$, where $\mathbf{U}_1$, $\mathbf{U}_2$ are two unimodular matrices in $k[\z_2]^{l\times l}$, and $\mathbf{D}={\rm diag}(\underbrace{h,\ldots,h}_{r}, \underbrace{1,\ldots,1}_{l-r})$. Then, ${\rm Im}(\mathbf{F}_1)={\rm Im}(\mathbf{F}_2)$.
 \end{theorem}

 \begin{proof}
  Let $\mathbf{F}_1 = \begin{bmatrix} \vec{u}_1^{\rm T},  \ldots, \vec{u}_l^{\rm T} \end{bmatrix}^{\rm T}$ and $\mathbf{F}_2 = \begin{bmatrix} \vec{v}_1^{\rm T} ,  \ldots,  \vec{v}_l^{\rm T} \end{bmatrix}^{\rm T}$, where $\vec{u}_1,\ldots,$ $\vec{u}_l,\vec{v}_1,\ldots,\vec{v}_l\in k[\z]^{1\times m}$. So, ${\rm Im}(\mathbf{F}_1)= \langle  \vec{u}_1, \ldots, \vec{u}_l\rangle$ and ${\rm Im}(\mathbf{F}_2)$ $= \langle  \vec{v}_1, \ldots, \vec{v}_l\rangle$.

  Let $\mathbf{F}_{01}= \mathbf{U}_1\mathbf{F}$ and $\mathbf{F}_{02}= \mathbf{U}_2\mathbf{F}$. Then $\mathbf{F}_{01}= \mathbf{D}\mathbf{F}_1$ and $\mathbf{F}_{02}= \mathbf{D}\mathbf{F}_2$. It follows that $\mathbf{F}_{01} = \begin{bmatrix} h\vec{u}_1^{\rm T},  \ldots, h\vec{u}_{r}^{\rm T},\vec{u}_{r+1}^{\rm T},\ldots,\vec{u}_l^{\rm T} \end{bmatrix}^{\rm T}$ and $\mathbf{F}_{02} = \begin{bmatrix} h\vec{v}_1^{\rm T},  \ldots, h\vec{v}_{r}^{\rm T},\vec{v}_{r+1}^{\rm T},\ldots,\vec{v}_l^{\rm T} \end{bmatrix}^{\rm T}$. Since $\mathbf{U}_1$ and $\mathbf{U}_2$ are two unimodular matrices in $k[\z_2]^{l\times l}$, we have $\mathbf{F}_{01}= \mathbf{U}_1\mathbf{U}_2^{-1}\mathbf{F}_{02}$. This implies that there exist polynomials $a_{i1},\ldots,a_{il}\in k[\z_2]$ such that
  \begin{equation*}\label{unique-eq-1}
   h\vec{u}_i = h\cdot(\sum_{j=1}^{r}a_{ij}\vec{v}_j) + \sum_{j=r+1}^{l}a_{ij}\vec{v}_j,
  \end{equation*}
  where $i=1,\ldots,r$. Then, for each $i$ setting $z_1$ of the above equation to $f(\z_2)$, we have
  $$a_{i(r+1)}\vec{v}_{r+1}(f,\z_2)+\cdots+ a_{il}\vec{v}_l(f,\z_2)= \vec{0}.$$
  As ${\rm rank}(\mathbf{F}(f,\z_2)) = l-r$ and ${\rm rank}(\mathbf{F}_{02}(f,\z_2)) = {\rm rank}(\mathbf{F}(f,\z_2))$, we have that $\vec{v}_{r+1}(f,\z_2),\ldots,$ $\vec{v}_l(f,\z_2)$ are $k[\z_2]$-linearly independent. This implies that $a_{i(r+1)} = \cdots = a_{il} = 0$. Hence,
  $$\vec{u}_i = a_{i1}\vec{v}_1+ \cdots + a_{ir}\vec{v}_{r},$$
  where $i=1,\ldots,r$. Obviously, $\vec{u}_j$ is a $k[\z]$-linear combination of $\vec{v}_1, \ldots, \vec{v}_l$, where $j = r+1,$ $\ldots,l$. As a consequence, $\langle  \vec{u}_1, \ldots, \vec{u}_l\rangle \subset \langle  \vec{v}_1, \ldots, \vec{v}_l\rangle$. We can use the same method to prove that $\langle  \vec{v}_1, \ldots, \vec{v}_l\rangle \subset \langle  \vec{u}_1, \ldots, \vec{u}_l\rangle$.

  Therefore, we have ${\rm Im}(\mathbf{F}_1) = {\rm Im}(\mathbf{F}_2)$.
 \end{proof}

 Based on Theorem \ref{main-theorem-uniqueness}, we can now derive the conclusion: the output $\F_1$ of Algorithm \ref{MF_Algorithm} is unique, i.e., ${\rm Im}(\mathbf{F}_1)$ is uniquely determined.

 \section{Examples}\label{sec_examples}

  We use two examples to illustrate the calculation process of Algorithm \ref{MF_Algorithm}. We first return to Example \ref{counter-example-1}.

  \begin{example}\label{counter-example-3}
  Let
  \[\mathbf{F} =
  \begin{bmatrix}
       -2z_1z_2^2+z_1^2z_3+z_2^2z_3-z_1z_3^2+z_2z_3^2   & z_1^3-z_2^3-z_1^2z_3+z_2z_3^2  & z_1z_2-z_2z_3 & z_2^2   \\
         -z_1z_2+z_3^2    &     -z_2^2+z_1z_3  & 0 & z_2
   \end{bmatrix}\]
  be a polynomial matrix in $\mathbb{C}[z_1,z_2,z_3]^{2\times 4}$, where $z_1>z_2>z_3$ and $\mathbb{C}$ is the complex field.

  It is easy to compute that $d_2(\mathbf{F})=z_2(z_1-z_3)$ and $d_1(\mathbf{F})=1$. Let $\mathbf{F}$, $h=z_1-z_3$ and $\prec_{z_2,z_3}$ be the inputs of Algorithm \ref{MF_Algorithm}, where $\prec_{z_2,z_3}$ is the degree reverse lexicographic order.

  Note that
  \[\mathbf{F}(z_3,z_2,z_3) =
  \begin{bmatrix}
       -z_2^2z_3+z_2z_3^2   & -z_2^3+z_2z_3^2  & 0 & z_2^2   \\
         -z_2z_3+z_3^2    &     -z_2^2+z_3^2  & 0 & z_2
   \end{bmatrix},\]
  ${\rm rank}(\mathbf{F}(z_3,z_2,z_3)) = 1$ and $r=1$. All the $1\times 1$ column reduced minors of $\mathbf{F}(z_3,z_2,z_3)$ are $z_2,1$. Since the reduced \gr basis of $\langle z_2,1 \rangle$ w.r.t. $\prec_{z_2,z_3}$ is $\{1\}$, $\F$ has a matrix factorization w.r.t. $h$.

  Let $W = {\rm Im}(\mathbf{F}(z_3,z_2,z_3))$. Then we compute a reduced \gr basis of the syzygy module of $W$, and obtain
  \[\mathbf{H} =
  \begin{bmatrix}
    1  &  -z_2
  \end{bmatrix}.\]
  It is easy to check that $\mathbf{H}$ is a ZLP matrix. $\mathbf{H}$ can be extended as the first row of a unimodular matrix
  \[\mathbf{U} =
  \begin{bmatrix}
       1   & -z_2   \\
       0   & 1
  \end{bmatrix}\]
  by using the package QUILLENSUSLIN. We extract $h$ from the first row of $\mathbf{U} \mathbf{F}$, and get
  \[\mathbf{U} \mathbf{F} = \mathbf{D}\mathbf{F}_1 =
    \begin{bmatrix}
    z_1 - z_3 & 0 \\
       0      & 1
   \end{bmatrix}
   \begin{bmatrix}
    z_1z_3-z_2^2   &  z_1^2-z_2z_3   & z_2 &  0   \\
    -z_1z_2+z_3^2  &  -z_2^2+z_1z_3  &  0  & z_2
   \end{bmatrix}.\]
  Then, $\mathbf{F}$ has a matrix factorization w.r.t. $h$:
  \[\mathbf{F} = \mathbf{G}_1 \mathbf{F}_1 = (\mathbf{U}^{-1} \mathbf{D})\mathbf{F}_1 =
    \begin{bmatrix}
    z_1 - z_3 & z_2 \\
       0      & 1
   \end{bmatrix}
   \begin{bmatrix}
    z_1z_3-z_2^2   &   z_1^2-z_2z_3  & z_2 &  0   \\
    -z_1z_2+z_3^2  &  -z_2^2+z_1z_3  &  0  & z_2
   \end{bmatrix},\]
  where ${\rm det}(\mathbf{G}_1)={\rm det}(\mathbf{U}^{-1} \mathbf{D}) = h$.

  At this moment, $d_2(\mathbf{F}_1) = z_2$. We reuse Algorithm \ref{MF_Algorithm} to judge whether $\mathbf{F}_1$ has a matrix factorization w.r.t. $z_2$. Note that
  \[\mathbf{F}_1(z_1,0,z_3) =
  \begin{bmatrix}
       z_1z_3   & z_1^2  & 0 & 0   \\
       z_3^2    & z_1z_3  & 0 & 0
   \end{bmatrix},\]
  ${\rm rank}(\mathbf{F}_1(z_1,0,z_3)) = 1$ and $r = 1$. All the $1\times 1$ column reduced minors of $\mathbf{F}_1(z_1,0,z_3)$ are $z_1,z_3$, and the reduced \gr basis $\mathcal{G}$ of $\langle z_1, z_3\rangle$ is $\{z_1,z_3\}$. Since $\mathcal{G} \neq \{ 1 \}$ and $r = 1$, $\mathbf{F}_1$ has no matrix factorizations w.r.t. $z_2$.
 \end{example}

 \begin{remark}
  In Example \ref{counter-example-3}, we can first judge whether  $\mathbf{F}$ has a matrix factorization w.r.t. $z_2$. Note that
  \[\mathbf{F}(z_1,0,z_3) =
  \begin{bmatrix}
       z_1z_3(z_1-z_3)   & z_1^2(z_1-z_3)  & 0 & 0   \\
         z_3^2    &     z_1z_3  & 0 & 0
   \end{bmatrix},\]
  ${\rm rank}(\mathbf{F}(z_1,0,z_3)) = 1$ and $r = 1$. All the $1\times 1$ column reduced minors of $\mathbf{F}(z_1,0,z_3)$ are $z_1(z_1-z_3),z_3$, and do not generate $k[z_1,z_3]$. This implies that $\mathbf{F}$ has no matrix factorizations w.r.t. $z_2$.

  According to the above calculations, we have the following conclusion: $\mathbf{F}$ has a matrix factorization w.r.t. $z_1-z_3$, but does not have a matrix factorization w.r.t. $z_2$.
 \end{remark}

 \begin{example}\label{counter-example-4}
  Let
  \[\mathbf{F} =
  \begin{bmatrix}
       z_1^2-z_1z_2   & z_2z_3+z_3^2+z_2+z_3  & -z_2z_3-z_2    \\
       z_1z_2-z_2^2  & -z_1z_3+z_2z_3  & z_1^3-z_1^2z_2+z_1z_2-z_2^2 \\
         0    &     z_2+z_3  &  -z_2
   \end{bmatrix}\]
  be a polynomial matrix in $\mathbb{C}[z_1,z_2,z_3]^{3\times 3}$, where $z_1>z_2>z_3$ and $\mathbb{C}$ is the complex field.

  It is easy to compute that $d_3(\mathbf{F})=-z_1(z_1-z_2)^2(z_1^2z_2+z_1^2z_3+z_2^2)$, $d_2(\mathbf{F})=z_1-z_2$ and $d_1(\mathbf{F})=1$. Let $\mathbf{F}$, $h=z_1-z_2$ and $\prec_{z_2,z_3}$ be the inputs of Algorithm \ref{MF_Algorithm}, where $\prec_{z_2,z_3}$ is the degree reverse lexicographic order.

  Note that
  \[\mathbf{F}(z_2,z_2,z_3) =
  \begin{bmatrix}
       0    & (z_2+z_3)(z_3+1)  & -z_2(z_3+1)    \\
       0    &     0     &    0    \\
       0    &  z_2+z_3  &  -z_2
   \end{bmatrix},\]
  ${\rm rank}(\mathbf{F}(z_2,z_2,z_3))=1$ and $r=2$. Obviously, all the $1\times 1$ column reduced minors of $\F(z_2,z_2,z_3)$ are $z_3+1, 1$. Since the reduced \gr basis of $\langle z_3+1, 1 \rangle$ w.r.t. $\prec_{z_2,z_3}$ is $\{1\}$, $\F$ has a matrix factorization w.r.t. $h^2$.

  Let $W = {\rm Im}(\mathbf{F}(z_2,z_2,z_3))$. Then we compute a reduced \gr basis of the syzygy module of $W$, and obtain
  \[\mathbf{H} =
  \begin{bmatrix}
       1   & 0  & -z_3-1   \\
       0   & 1  &  0
  \end{bmatrix}.\]
  It is easy to check that the reduced \gr basis of all the $2\times 2$ minors of $\mathbf{H}$ w.r.t. $\prec_{z_2,z_3}$ is $\mathcal{G} = \{ 1 \}$. Then, $\mathbf{H}$ is a ZLP matrix. We use the package QUILLENSUSLIN to construct a unimodular matrix
  \[\mathbf{U} =
  \begin{bmatrix}
       1   & 0  & -z_3-1   \\
       0   & 1  &  0       \\
       0   & 0  &  1
  \end{bmatrix}\]
  such that $\mathbf{H}$ is the first $2$ rows of $\mathbf{U}$. We extract $h$ from the first $2$ rows of $\mathbf{U} \mathbf{F}$, and get
  \[\mathbf{U} \mathbf{F} = \mathbf{D}\mathbf{F}_1 =
    \begin{bmatrix}
    z_1 - z_2 & 0         &  0   \\
       0      & z_1 - z_2 &  0   \\
       0      &    0      &  1
   \end{bmatrix}
   \begin{bmatrix}
    z_1   &   0       & 0           \\
    z_2   &  -z_3      &  z_1^2+z_2  \\
    0     &  z_2+z_3  &  -z_2
   \end{bmatrix}.\]
  Then, we obtain a matrix factorization of $\mathbf{F}$ w.r.t. $h^2$:

  \[\mathbf{F} = \mathbf{G}_1 \mathbf{F}_1= (\mathbf{U}^{-1} \mathbf{D})\mathbf{F}_1 =
    \begin{bmatrix}
    z_1 - z_2 & 0         &  z_3+1   \\
       0      & z_1 - z_2 &  0   \\
       0      &    0      &  1
   \end{bmatrix}
   \begin{bmatrix}
    z_1   &   0       & 0           \\
    z_2   &  -z_3      &  z_1^2+z_2  \\
    0     &  z_2+z_3  &  -z_2
   \end{bmatrix},\]
  where ${\rm det}(\mathbf{G}_1)={\rm det}(\mathbf{U}^{-1} \mathbf{D}) = h^2$.

  At this moment, $d_3(\mathbf{F}_1)=-z_1(z_1^2z_2+z_1^2z_3+z_2^2)$. We reuse Algorithm \ref{MF_Algorithm} to judge whether $\mathbf{F}_1$ has a matrix factorization w.r.t. $z_1$. Similarly, we obtain
  \[\mathbf{F}_1 = \mathbf{G}_2 \mathbf{F}_2 =
    \begin{bmatrix}
       z_1    &    0      &  0  \\
       0      &    1      &  0   \\
       0      &    0      &  1
   \end{bmatrix}
   \begin{bmatrix}
      1   &   0       & 0           \\
    z_2   &  -z_3      &  z_1^2+z_2  \\
    0     &  z_2+z_3  &  -z_2
   \end{bmatrix},\]
  where ${\rm det}(\mathbf{G}_2)=z_1$.

  Therefore, we obtain a matrix factorization of $\mathbf{F}$ w.r.t. $z_1(z_1-z_2)^2$, i.e.,
  \[ \mathbf{F} = \mathbf{G} \mathbf{F}_2 = (\mathbf{G}_1\mathbf{G}_2)\mathbf{F}_2 =
    \begin{bmatrix}
    z_1(z_1 - z_2) & 0         &  z_3+1   \\
       0      & z_1 - z_2 &  0   \\
       0      &    0      &  1
   \end{bmatrix}
   \begin{bmatrix}
    1   &   0       & 0           \\
    z_2   &  -z_3      &  z_1^2+z_2  \\
    0     &  z_2+z_3  &  -z_2
   \end{bmatrix},\]
  where ${\rm det}(\mathbf{G})=z_1(z_1-z_2)^2$.
 \end{example}

 \begin{remark}
  In Example \ref{counter-example-4}, we can first judge whether  $\mathbf{F}$ has a matrix factorization w.r.t. $z_1$. Note that
  \[\mathbf{F}(0,z_2,z_3) =
  \begin{bmatrix}
       0    & (z_2+z_3)(z_3+1)  & -z_2(z_3+1)    \\
    -z_2^2  &      z_2z_3  & -z_2^2 \\
       0    &     z_2+z_3  &  -z_2
   \end{bmatrix},\]
  ${\rm rank}(\mathbf{F}(0,z_2,z_3)) = 2$ and $r = 1$. All the $2\times 2$ column reduced minors of $\mathbf{F}(0,z_2,z_3)$ are $z_3+1,1$, and generate $k[z_2,z_3]$. This implies that $\mathbf{F}$ has a matrix factorization w.r.t. $z_1$.

  According to the above calculations, we have the following conclusion: $\mathbf{F}$ has a matrix factorization w.r.t. $z_1$, $z_1-z_2$, $z_1(z_1-z_2)$, $(z_1-z_2)^2$ and $z_1(z_1-z_2)^2$, respectively.
 \end{remark}

\section{Concluding remarks}\label{sec_conclusions}

 In this paper, we point out two directions of research in which multivariate polynomial matrices have been explored. The first is concerned with the factorization problem for a class of multivariate polynomial matrices, and the second direction is devoted to the investigation of the equivalence problem of a square polynomial matrix and a diagonal matrix.

 The main contributions of this paper include: 1) some new factorization criteria are given to factorize $\mathbf{F}\in \mathcal{M}$ w.r.t. $h^r$, and the relationships among all existed factorization criteria have been studied; 2) a necessary and sufficient condition is proposed to judge whether a square polynomial matrix with the determinant being $h^r$ is equivalent to the diagonal matrix ${\rm diag}(h,\ldots,h,1,\ldots,1)$; 3) based on new criteria, a factorization algorithm is given and the output of the algorithm is proved to be unique; 4) the algorithm is implemented on Maple, and two examples are given to illustrate the effectiveness of the algorithm.

 A sufficient condition is obtained for the existence of a matrix factorization of $\F$ w.r.t. $h^r~(1< r < l)$. At this moment, how to establish a necessary and sufficient condition for $\mathbf{F}$ admitting a matrix factorization w.r.t. $h^r$ is the question that remains for further investigation.

\section*{Acknowledgments}

 This research was supported by the CAS Key Project QYZDJ-SSW-SYS022.

\bibliographystyle{elsarticle-harv}

\bibliography{20-mf}

\begin{thebibliography}{57}
\expandafter\ifx\csname natexlab\endcsname\relax\def\natexlab#1{#1}\fi
\expandafter\ifx\csname url\endcsname\relax
  \def\url#1{\texttt{#1}}\fi
\expandafter\ifx\csname urlprefix\endcsname\relax\def\urlprefix{URL }\fi

\bibitem[{Bose(1982)}]{Bose1982}
Bose, N., 1982. Applied Multidimensional Systems Theory. Van Nostrand Reinhold
  Co., New York.

\bibitem[{Bose et~al.(2003)Bose, Buchberger, and Guiver}]{Bose2003}
Bose, N., Buchberger, B., Guiver, J., 2003. Multidimensional Systems Theory and
  Applications. Kluwer Academic Publishers, Dordrecht, The Netherlands.

\bibitem[{Boudellioua(2012)}]{Boudellioua2012Com}
Boudellioua, M., 2012. Computation of the {S}mith form for multivariate
  polynomial matrices using {M}aple. American Journal of Computational
  Mathematics 2~(1), 21--26.

\bibitem[{Boudellioua(2013)}]{Boudellioua2013Further}
Boudellioua, M., 2013. Further results on the equivalence to {S}mith form of
  multivariate polynomial matrices. Control and Cybernetics 42~(2), 543--551.

\bibitem[{Boudellioua(2014)}]{Boudellioua2014}
Boudellioua, M., 2014. Computation of a canonical form for linear 2-{D}
  systems. International Journal of Computational Mathematics 2014~(487465),
  1--6.

\bibitem[{Boudellioua and Quadrat(2010)}]{Boudellioua2010Serre}
Boudellioua, M., Quadrat, A., 2010. Serre's reduction of linear function
  systems. Mathematics in Computer Science 4~(2-3), 289--312.

\bibitem[{Charoenlarpnopparut and
  Bose(1999)}]{Charoenlarpnopparut1999Multidimensional}
Charoenlarpnopparut, C., Bose, N., 1999. Multidimensional {FIR} filter bank
  design using \gr bases. IEEE Transactions on Circuits and Systems II: Analog
  and Digital Signal Processing 46~(12), 1475--1486.

\bibitem[{Cluzeau and Quadrat(2008)}]{Cluzeau2008F}
Cluzeau, T., Quadrat, A., 2008. Factoring and decomposing a class of linear
  functional systems. Linear Algebra and Its Applications 428, 324--381.

\bibitem[{Cluzeau and Quadrat(2013)}]{Cluzeau2013Is}
Cluzeau, T., Quadrat, A., 2013. Isomorphisms and {S}erre's reduction of linear
  systems. In: Proceedings of the 8th International Workshop on
  Multidimensional Systems. VDE, Erlangen, Germany, pp. 1--6.

\bibitem[{Cluzeau and Quadrat(2015)}]{Cluzeau2015A}
Cluzeau, T., Quadrat, A., 2015. A new insight into {S}erre's reduction problem.
  Linear Algebra and its Applications 483, 40--100.

\bibitem[{Cox et~al.(2005)Cox, Little, and O'shea}]{Cox2005Using}
Cox, D., Little, J., O'shea, D., 2005. Using Algebraic Geometry. Graduate Texts
  in Mathematics (Second Edition). Springer, New York.

\bibitem[{Eisenbud(2013)}]{Eisenbud2013}
Eisenbud, D., 2013. Commutative Algebra: with a view toward algebraic geometry.
  New York: Springer.

\bibitem[{Fabia\'{n}ska and Quadrat(2007)}]{Fabianska2007Applications}
Fabia\'{n}ska, A., Quadrat, A., 2007. Applications of the Quillen-Suslin
  theorem to multidimensional systems theory. Vol.~3 of In: Park, H.,
  Regensburger, G. (Eds.), \gr Bases in Control Theory and Signal Processing,
  Radon Series on Computational and Applied Mathematics. Walter de Gruyter.

\bibitem[{Frost and Boudellioua(1986)}]{Frost1986Some}
Frost, M., Boudellioua, M., 1986. Some further results concerning matrices with
  elements in a polynomial ring. International Journal of Control 43~(5),
  1543--1555.

\bibitem[{Frost and Storey(1978)}]{Frost1978Equivalence}
Frost, M., Storey, C., 1978. Equivalence of a matrix over {R}$[s,z]$ with its
  {S}mith form. International Journal of Control 28~(5), 665--671.

\bibitem[{Greuel and Pfister(2002)}]{Greuel2002A}
Greuel, G., Pfister, G., 2002. A SINGULAR Introduction to Commutative Algebra.
  Springer-Verlag.

\bibitem[{Guan et~al.(2018)Guan, Li, and Ouyang}]{Guan2018}
Guan, J., Li, W., Ouyang, B., 2018. On rank factorizations and factor prime
  factorizations for multivariate polynomial matrices. Journal of Systems
  Science and Complexity 31~(6), 1647--1658.

\bibitem[{Guan et~al.(2019)Guan, Li, and Ouyang}]{Guan2019}
Guan, J., Li, W., Ouyang, B., 2019. On minor prime factorizations for
  multivariate polynomial matrices. Multidimensional Systems and Signal
  Processing 30, 493--502.

\bibitem[{Guiver and Bose(1982)}]{Guiver1982Polynomial}
Guiver, J., Bose, N., 1982. Polynomial matrix primitive factorization over
  arbitrary coefficient field and related results. IEEE Transactions on
  Circuits and Systems 29~(10), 649--657.

\bibitem[{Kailath(1993)}]{Kailath1980}
Kailath, T., 1993. Linear Systems. Englewood Cliffs, NJ: Prentice Hall.

\bibitem[{Kung et~al.(1977)Kung, Levy, Morf, and Kailath}]{Kung1977New}
Kung, S., Levy, B., Morf, M., Kailath, T., 1977. New results in 2-{D} systems
  theory, part {II}: 2-{D} state-space models-realization and the notions of
  controllability, observability, and minimality. In: Proceedings of the IEEE.
  Vol.~65. pp. 945--961.

\bibitem[{Lee and Zak(1983)}]{Lee1983Smith}
Lee, E., Zak, S., 1983. Smith forms over {R}$[z_1,z_2]$. IEEE Transactions on
  Automatic Control 28~(1), 115--118.

\bibitem[{Li et~al.(2017)Li, Liu, and Zheng}]{Li2017On}
Li, D., Liu, J., Zheng, L., 2017. On the equivalence of multivariate polynomial
  matrices. Multidimensional Systems and Signal Processing 28, 225--235.

\bibitem[{Lin(1988)}]{Lin1988On}
Lin, Z., 1988. On matrix fraction descriptions of multivariable linear $n$-{D}
  systems. IEEE Transactions on Circuits and Systems 35~(10), 1317--1322.

\bibitem[{Lin(1999{\natexlab{a}})}]{Lin1999Notes}
Lin, Z., 1999{\natexlab{a}}. Notes on $n$-{D} polynomial matrix factorizations.
  Multidimensional Systems and Signal Processing 10~(4), 379--393.

\bibitem[{Lin(1999{\natexlab{b}})}]{Lin1999On}
Lin, Z., 1999{\natexlab{b}}. On syzygy modules for polynomial matrices. Linear
  Algebra and Its Applications 298~(1-3), 73--86.

\bibitem[{Lin(2001)}]{Lin2001Further}
Lin, Z., 2001. Further results on $n$-{D} polynomial matrix factorizations.
  Multidimensional Systems and Signal Processing 12~(2), 199--208.

\bibitem[{Lin and Bose(2001)}]{Lin2001A}
Lin, Z., Bose, N., 2001. A generalization of {S}erre's conjecture and some
  related issues. Linear Algebra and Its Applications 338, 125--138.

\bibitem[{Lin et~al.(2006)Lin, Boudellioua, and Xu}]{Lin2006On}
Lin, Z., Boudellioua, M., Xu, L., 2006. On the equivalence and factorization of
  multivariate polynomial matrices. In: Proceeding of ISCAS. Greece, pp.
  4911--4914.

\bibitem[{Lin et~al.(2005)Lin, Li, and Fan}]{Lin2005On}
Lin, Z., Li, X., Fan, H., 2005. On minor prime factorizations for $n$-{D}
  polynomial matrices. IEEE Transactions on Circuits and Systems II: Express
  Briefs 52~(9), 568--571.

\bibitem[{Lin et~al.(2008)Lin, Xu, and Bose}]{Lin2008ATutorial}
Lin, Z., Xu, L., Bose, N., 2008. A tutorial on \gr bases with applications in
  signals and systems. IEEE Transactions on Circuits and Systems I: Regular
  Papers 55~(1), 445--461.

\bibitem[{Lin et~al.(2001)Lin, Ying, and Xu}]{Lin2001Factorizations}
Lin, Z., Ying, J., Xu, L., 2001. Factorizations for $n$-{D} polynomial
  matrices. Circuits, Systems, and Signal Processing 20~(6), 601--618.

\bibitem[{Liu et~al.(2011)Liu, Li, and Wang}]{Liu2011On}
Liu, J., Li, D., Wang, M., 2011. On general factorizations for $n$-{D}
  polynomial matrices. Circuits Systems and Signal Processing 30~(3), 553--566.

\bibitem[{Liu et~al.(2014)Liu, Li, and Zheng}]{Liu2014The}
Liu, J., Li, D., Zheng, L., 2014. The {L}in-{B}ose problem. IEEE Transactions
  on Circuits and Systems II: Express Briefs 61~(1), 41--43.

\bibitem[{Liu and Wang(2010)}]{Liu2010Notes}
Liu, J., Wang, M., 2010. Notes on factor prime factorizations for $n$-{D}
  polynomial matrices. Multidimensional Systems and Signal Processing 21~(1),
  87--97.

\bibitem[{Liu and Wang(2013)}]{Liu2013New}
Liu, J., Wang, M., 2013. New results on multivariate polynomial matrix
  factorizations. Linear Algebra and Its Applications 438~(1), 87--95.

\bibitem[{Liu and Wang(2015)}]{Liu2015Further}
Liu, J., Wang, M., 2015. Further remarks on multivariate polynomial matrix
  factorizations. Linear Algebra and Its Applications 465, 204--213.

\bibitem[{Logar and Sturmfels(1992)}]{Logar1992Algorithms}
Logar, A., Sturmfels, B., 1992. Algorithms for the {Q}uillen-{S}uslin theorem.
  Journal of Algebra 145~(1), 231--239.

\bibitem[{Lu et~al.(2017)Lu, Ma, and Wang}]{Lu2017}
Lu, D., Ma, X., Wang, D., 2017. A new algorithm for general factorizations of
  multivariate polynomial matrices. In: proceedings of 42nd ISSAC. pp.
  277--284.

\bibitem[{Lu et~al.(2020{\natexlab{a}})Lu, Wang, and
  Xiao}]{Lu2019Factorizations}
Lu, D., Wang, D., Xiao, F., 2020{\natexlab{a}}. Factorizations for a class of
  multivariate polynomial matrices. Multidimensional Systems and Signal
  Processing 31, 989--1004.

\bibitem[{Lu et~al.(2020{\natexlab{b}})Lu, Wang, and Xiao}]{Lu2020Further}
Lu, D., Wang, D., Xiao, F., 2020{\natexlab{b}}. Further results on the
  factorization and equivalence for multivariate polynomial matrices. In:
  proceedings of 45th ISSAC. pp. 328--335.

\bibitem[{Morf et~al.(1977)Morf, Levy, and Kung}]{Morf1977New}
Morf, M., Levy, B., Kung, S., 1977. New results in 2-{D} systems theory, part
  {I}: 2-{D} polynomial matrices, factorization, and coprimeness. In:
  Proceedings of the IEEE. Vol.~65. pp. 861--872.

\bibitem[{Park(1995)}]{park1995}
Park, H., 1995. A computational theory of {L}aurent polynomial rings and
  multidimensional {FIR} systems. Ph.D. thesis, University of California at
  Berkeley.

\bibitem[{Pommaret(2001)}]{Pommaret2001Solving}
Pommaret, J., 2001. Solving {B}ose conjecture on linear multidimensional
  systems. In: Proceeding of European Control Conference. IEEE, pp. 1653--1655.

\bibitem[{Pugh et~al.(1998)Pugh, Mcinerney, Boudellioua, Johnson, and
  Hayton}]{Pugh1998A}
Pugh, A., Mcinerney, S., Boudellioua, M., Johnson, D., Hayton, G., 1998. A
  transformaition for 2-{D} linear systems and a generalization of a theorem of
  rosenbrock. International Journal of Control 71~(3), 491--503.

\bibitem[{Quillen(1976)}]{Quillen1976Projective}
Quillen, D., 1976. Projective modules over polynomial rings. Inventiones
  mathematicae 36, 167--171.

\bibitem[{Rosenbrock(1970)}]{Rosenbrock1970}
Rosenbrock, H., 1970. State-{S}pace and Multivariable Theory. London: Nelson.

\bibitem[{Serre(1955)}]{serre1955faisceaux}
Serre, J., 1955. Faisceaux alg{\'e}briques coh{\'e}rents. Annals of Mathematics
  (Second Series) 61~(2), 197--278.

\bibitem[{Srinivas(2004)}]{Srinivas2004A}
Srinivas, V., 2004. A generalized {S}erre problem. Journal of Algebra 278~(2),
  621--627.

\bibitem[{Strang(1980)}]{Strang2010Linear}
Strang, G., 1980. Linear Algebra and Its Applications (Second Edition).
  Academic Press.

\bibitem[{Sule(1994)}]{Sule1994Feed}
Sule, V., 1994. Feedback stabilization over commutative rings: the matrix case.
  SIAM Journal on Control and Optimization 32~(6), 1675--1695.

\bibitem[{Suslin(1976)}]{Suslin1976Projective}
Suslin, A., 1976. Projective modules over polynomial rings are free. Soviet
  mathematics - Doklady 17, 1160--1165.

\bibitem[{Wang(2007)}]{Mingsheng2007On}
Wang, M., 2007. On factor prime factorization for $n$-{D} polynomial matrices.
  IEEE Transactions on Circuits and Systems I: Regular Papers 54~(6),
  1398--1405.

\bibitem[{Wang and Feng(2004)}]{Wang2004On}
Wang, M., Feng, D., 2004. On {L}in-{B}ose problem. Linear Algebra and Its
  Applications 390, 279--285.

\bibitem[{Wang and Kwong(2005)}]{Mingsheng2005On}
Wang, M., Kwong, C., 2005. On multivariate polynomial matrix factorization
  problems. Mathematics of Control, Signals and Systems 17~(4), 297--311.

\bibitem[{Youla and Gnavi(1979)}]{Youla1979Notes}
Youla, D., Gnavi, G., 1979. Notes on $n$-dimensional system theory. IEEE
  Transactions on Circuits and Systems 26~(2), 105--111.

\bibitem[{Youla and Pickel(1984)}]{Youla1984The}
Youla, D., Pickel, P., 1984. The {Q}uillen-{S}uslin theorem and the structure
  of $n$-dimensional elementary polynomial matrices. IEEE Transactions on
  Circuits and Systems 31~(6), 513--518.

\end{thebibliography}

\end{document}